\documentclass[conference]{IEEEtran}

\ifCLASSINFOpdf
\else
\fi

\PassOptionsToPackage{prologue, x11names}{xcolor}
\usepackage[table]{xcolor}
\usepackage{amssymb}
\usepackage{amsmath,amsfonts}
\usepackage{graphics}
\usepackage{textcomp}
\usepackage{xcolor}
\usepackage{hyperref}

\usepackage{tikz}
\usepackage{multirow}
\usepackage{amsthm}
\newtheorem{theorem}{Theorem}

\newtheorem{lemma}{Lemma}
\usepackage{paralist}
\usepackage[ruled,vlined]{algorithm2e}
\usepackage{url}

\usepackage{threeparttable}
\usepackage{booktabs}
\usepackage{makecell}
\usepackage{wasysym}
\usepackage{enumitem}
\usepackage{comment}
\usepackage{adjustbox}
\usepackage{pifont}

\def\Rbb{\mathbb{R}}
\def\Ebb{\mathbb{E}}
\def\local{\mathsf{local}}
\def\globalsf{\mathsf{global}}
\def\jsequ{S_{\mathsf{Jac}}}
\newcommand{\novel}{\textcolor{blue}{\(\blacklozenge\)}}
\newcommand{\systemNumber}{FedAVG-N\xspace}
\newcommand{\system}{\textsc{Entente}\xspace}
\newcommand\zl[1]{{\color{purple}{\textbf{\{ZL: {\em#1}\}}}}}
\newcommand\jxu[1]{{\color{orange}{\textbf{\{Jason: {\em#1}\}}}}}

\newcommand{\zztitle}[1]{\noindent\textbf{#1 }}
\newcommand{\ignore}[1]{}

\usepackage{tablefootnote}

\definecolor{colorbest}{named}{DodgerBlue2}
\definecolor{colorsecond}{named}{SkyBlue1}
\definecolor{colortext}{named}{white}

\newcommand{\greenup}{$\mathrel{\hspace{0.04cm}}$\rotatebox{90}{\color{green}\ding{225}}}
\newcommand{\reddown}{\textcolor{purple}{$\blacktriangledown$}}

\newcommand{\tbspace}{\mathrel{\hspace{0.18cm}}}

\newcommand{\revisednew}[1]{{\color{black}{{#1}}}}

\begin{document}

\thispagestyle{plain}
\pagestyle{plain}

\title{\system: Cross-silo Intrusion Detection on Network Log Graphs with Federated Learning}

\author{%
  \IEEEauthorblockN{%
    Jiacen Xu\IEEEauthorrefmark{1}\textsuperscript{\textsection},
    Chenang Li\IEEEauthorrefmark{2},
    Yu Zheng\IEEEauthorrefmark{2} and
    Zhou Li\IEEEauthorrefmark{2}%
  }%
  \IEEEauthorblockA{\IEEEauthorrefmark{1} Microsoft}%
  \IEEEauthorblockA{\IEEEauthorrefmark{2} University of California, Irvine}%
}

\IEEEoverridecommandlockouts
\makeatletter\def\@IEEEpubidpullup{6.5\baselineskip}\makeatother
\IEEEpubid{\parbox{\columnwidth}{
		Network and Distributed System Security (NDSS) Symposium 2026\\
		23 - 27 February 2026 , San Diego, CA, USA\\
		ISBN 979-8-9919276-8-0\\  
		https://dx.doi.org/10.14722/ndss.2026.230093\\
		www.ndss-symposium.org
}
\hspace{\columnsep}\makebox[\columnwidth]{}}

\maketitle

\begingroup\renewcommand\thefootnote{\textsection}
\footnotetext{This work was done when the author was a PhD student at UC Irvine.}
\endgroup

\begin{abstract}
Graph-based Network Intrusion Detection Systems (GNIDS) have gained significant momentum in detecting sophisticated cyber-attacks, such as Advanced Persistent Threats (APTs), within and across organizational boundaries.
Though achieving satisfying detection accuracy and demonstrating adaptability to ever-changing attacks and normal patterns, existing GNIDS  predominantly assume a centralized data setting. 
However, flexible data collection is not always realistic or achievable due to increasing constraints from privacy regulations and operational limitations.

We argue that the practical development of GNIDS requires accounting for distributed collection settings and we leverage Federated Learning (FL) as a viable paradigm to address this prominent challenge.
We observe that naively applying FL to GNIDS is unlikely to be effective, due to issues like
graph heterogeneity over clients and the diverse design choices taken by different GNIDS. 
We address these issues with a set of novel techniques tailored to the graph datasets, including reference graph synthesis, graph sketching and adaptive contribution scaling, eventually developing a new system \system. 
By leveraging the domain knowledge, \system can achieve effectiveness, scalability and robustness simultaneously.
Empirical evaluation on the large-scale LANL, OpTC and Pivoting datasets shows that \system outperforms the SOTA FL baselines.
We also evaluate \system under FL poisoning attacks tailored to the GNIDS setting, showing the robustness by bounding the attack success rate to low values.
Overall, our study suggests a promising direction to build cross-silo GNIDS.
\end{abstract}

\IEEEpeerreviewmaketitle

\section{Introduction}
\label{sec:intro}

The techniques and scale of modern cyber-attacks are evolving at a rapid pace. More high-profile security breaches are observed against large organizations nowadays. One prominent attack strategy is Advanced Persistent Threat (APT)~\cite{killchain}, which 
establishes multiple attack stages and infiltrates multiple organizational assets through techniques like lateral movement. 
As a popular countermeasure, many organizations %
collect network logs (e.g., firewall and proxy logs) and perform intrusion detection on them~\cite{li2016operational}. To more precisely capture the distinctive network communication patterns of the attack, a promising approach is to model the network logs as a graph and apply graph-based algorithms to detect abnormal entities, interactions or communities. We term such system \textit{Graph-based Network Intrusion Detection System (GNIDS)}, and we observe that the recent works~\cite{xu2024understanding,king2023euler,khoury2024jbeil,paudel2022pikachu,liu2019log2vec,bowman2020detecting,cheng2021step,sun2022hetglm,rabbani2024graph,qiu20233d} prefer advanced graph models like graph autoencoder (GAE)~\cite{kipf2016variational} to build their systems, showing much higher detection accuracy over the traditional NIDS and capabilities of detecting sophisticated attacks like lateral movement~\cite{khoury2024jbeil}.

\zztitle{Regulation compliance concerns for GNIDS.}
Given the sensitive nature of network logs, such as revealing communication patterns of employees and organizations~\cite{imana2021institutional}, privacy regulations have to be followed when training GNIDS models with logs from multiple regions. 
For example, Menges et al. state that SIEM needs to be compliant with Europe’s General Data Protection Regulation (GDPR)~\cite{menges2021towards}, which covers the data processing and transfer ``within and between private companies and/or public bodies in the European member states''. Outside of Europe,
Singapore’s Personal Data Protection Act (PDPA) prohibits using data for purposes beyond its original intent without explicit individual consent~\cite{pdpa}, creating barriers for training models on network logs. 

\revisednew{
In fact, the data collection capabilities of a cyber-security company have already been restricted under data privacy regulations.  
For instance, Palo Alto Networks (PANW) offers a Strata Logging Service that enables enterprises to send on-premise firewall logs to the cloud for centralized analysis and management. 
However, PANW explicitly states that if regulations mandate data residency, customers must ensure that logs are stored in region-appropriate cloud instances to comply with jurisdictional boundaries~\cite{panw-strata}.
In such cases, logs cannot be aggregated across regions, making centralized analysis on logs impractical. 
Similarly, Microsoft’s Windows Defender XDR stores customer data, such as alerts, in regional Microsoft Azure data centers (e.g., EU, UK, US), and its documentation confirms that cloud tenants cannot be relocated across regions once created~\cite{ms-xdr}.
Our communication with product representatives verified that cross-region data centralization is unsupported.
}

\zztitle{Federated Learning for GNIDS and \system.}
The aforementioned compliance issue calls for a new paradigm that allows the development of GNIDS over geographically distributed logs while aligning with various privacy regulations. 
One promising solution is Federated Learning (FL), which has gained prominent attention from academia and industry~\cite{yang2019federated}.
In essence, FL allows the individual data owners (e.g., a device owner or an organization) to keep their data on premise and jointly train a global model by exchanging parameters of local models. 

Given its successes in addressing privacy concerns of data collection~\cite{yang2019federated},
we pivot the research of developing a practical FL-empowered GNIDS and term our new system \system. 

We argue that \system should satisfy three main design goals: 
\textit{effectiveness} (similar effectiveness as the GNIDS trained on the entire dataset), \textit{scalability} (the overhead introduced by the FL mechanism should be small and a large number of FL clients should be supported) and \textit{robustness} (maintaining detection accuracy against attackers who compromise the FL procedure). Since the data to be processed by GNIDS usually have imbalanced classes (e.g., malicious events are far less than the normal events) and non-IID (not independent and identically distributed) across FL clients, based on our survey, unfortunately \textit{none} of the existing FL methods are able to achieve these goals altogether. For instance, sharing neighborhood information through Fully Homomorphic Encryption (FHE) could mitigate the accuracy loss on the non-IID clients' graphs~\cite{yao2022fedgcn}, but doing so will introduce prominent overhead. Clipping the clients' contributions can curb the poisoning attack~\cite{sun2019can}, but the training convergence and model accuracy will be affected adversely~\cite{icml/ZhangCH0Y22}. 

\textit{Is it possible to build a federated GNIDS that achieves effectiveness, scalability and robustness all together, rather than sacrificing one goal for another?}

We answer firmly to this question by designing a new Federated Graph Learning (FGL) protocol for GNIDS. 1) We found that by sharing a small piece of clients' information, i.e., the aggregated node number, the central parameter server can effectively initialize the clients' initial weights, and mitigate the impact of non-IID clients. Such information is often already accessible within an organization so no extra privacy leakage will be introduced. 
We instantiate this idea with a new FL bootstrapping stage based on \textit{reference graph synthesis} and \textit{graph sketching}, which only involve lightweight computation on the parameter server and FL clients. 
2) We found that each client can self-adjust its contribution based on the divergence between client-to-global models, and we developed a new technique termed \textit{adaptive contribution scaling (ACS)} to instantiate this idea. 
3) The attacker needs to scale up the model updates to effectively poison the trained global model. 
Since the clients' weights are adjusted under ACS already, we can bound the model updates by tweaking ACS. Interestingly, such a combination enables \textit{dynamic clipping}, which can address the limitation of static clipping~\cite{icml/ZhangCH0Y22}. 
Through theoretical analysis,
we \textit{formally prove} that the iteration-wise difference shifting is bounded under ACS, and convergence rate is still bounded.  
This new theoretical result suggests our protocol could be useful for other FGL applications aiming to achieve robustness on non-IID clients.

\zztitle{Evaluation of \system.}
We conduct an extensive evaluation on \system, focusing on its effectiveness in detecting abnormal interactions between entities. We adapt \system to two exemplar GNIDS, namely Euler~\cite{king2023euler} and Jbeil~\cite{khoury2024jbeil}, as they embody quite different designs. 
We choose three real-world large-scale log datasets, OpTC~\cite{optc}, LANL Cyber1 (or LANL)~\cite{lanl-ds-15}  and Pivoting~\cite{apruzzese2017detection} for evaluation and simulate different client numbers.
In summary, \system can boost the performance of both GNIDS models on all the datasets.
On OpTC, \system outperforms all the other baseline FL methods
and \textit{even the non-FL version} (the GNIDS is trained using all data) with over 0.1 increase of average precision (AP). On LANL and Pivoting, when link prediction is conducted by Jbeil, high AP and AUC can be reached by \system (over 0.9 in most cases), for both transductive learning and inductive learning modes. On LANL when Euler is used, given only hundreds of redteam events are used for edge classification, the AP is low for all FL methods, but \system still outperforms the other methods in most cases.

We also consider the robustness of \system under adaptive attacks and consider model poisoning~\cite{bagdasaryan2020backdoor} as the main threat. We develop a \textit{new} poisoning attack that scales up the model updates~\cite{bagdasaryan2020backdoor} and adds covering edges~\cite{xu2023cover} to the GNIDS setting, by replaying \ignore{``copying'' \jxu{I want to use replay}} malicious edges from the testing period to the training period. 
Since \system integrates norm bounding when adjusting clients' weights, the attack success rate is bounded to a very low rate, e.g., less than 10\% when attacking Euler+LANL. Without norm bounding, not only does attack success rate increase, but the FL training process might not even finish when the attacker scales the model updates by a very large ratio.

Overall, our study shows promise in addressing the data sharing concerns in building GNIDS in practice.

\zztitle{Contributions.} We summarize the contributions of this paper as follows:

\begin{itemize}[noitemsep,nolistsep]
    \item We propose a new system \system that can train a GNIDS model without requesting data sharing among departments/ organizations, under the framework of FL.
    \item We address the threats like non-iid client graphs and adaptive attackers with novel techniques like reference graph synthesis, graph sketching and adaptive contribution scaling. We also formally prove that the convergence rate is bounded. 
    \item We conduct the extensive evaluation using large-scale log datasets (LANL, OpTC and Pivoting), and the result shows \system can outperform other baselines in most cases.
    \item We release the code at \url{https://github.com/uci-dsp-lab/ENTENTE}. 
\end{itemize}

\section{Background}
\label{sec:background}

\subsection{Graph-based Network Intrusion Detection Systems (GNIDS)}
\label{subsec:gsa}

Network logs collected by 
network appliances like firewalls and proxies have been extensively leveraged to detect various cyber-attacks, including APT attacks~\cite{killchain}. 
Many graph-based approaches have been developed in recent years and we term them Graph-based Network Intrusion Detection Systems (GNIDS).
At the high level, for each log entry, the GNIDS extracts the subject and object fields (e.g., host) as nodes, and fills the edge attributes using the other fields (e.g., the instruction contained in the network packet). In Figure~\ref{fig:optc}, we illustrate an example of a graph generated from network logs collected by different organizations (or clients). 

\begin{figure}[h]
    \centering
    \includegraphics[width=0.45\textwidth,page=3]{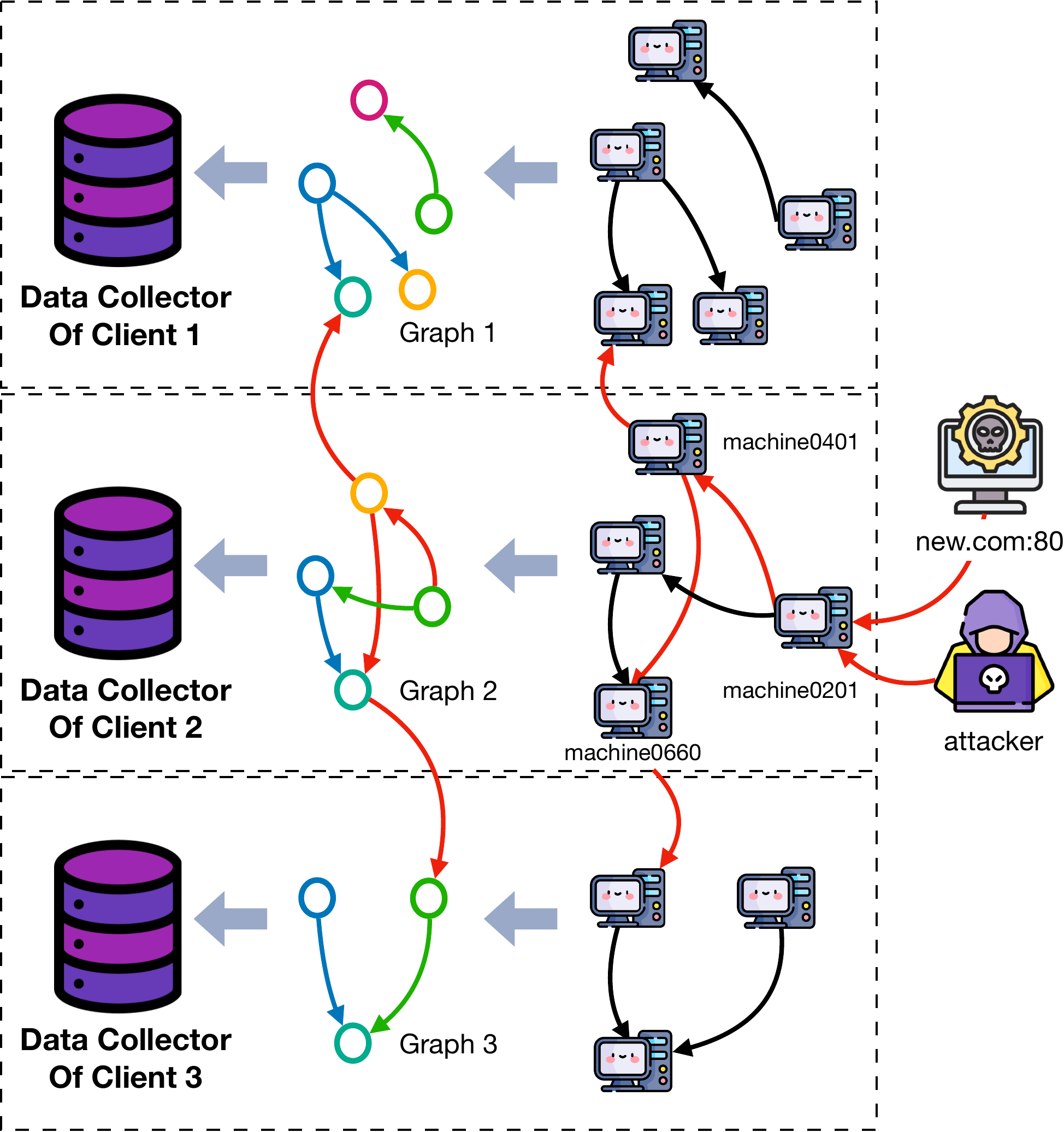}
    \caption{An example summarized from the day 1 attack campaign in the OpTC dataset~\cite{optc}.  The attacker first connects to \texttt{machine0201} and downloads the PowerShell attack tool. Then it pivots to machine \texttt{machine0401} and \texttt{machine0660} with the Windows WMI command. Finally, the attack spreads to other machines. The attacked machines can belong to multiple organizations (or clients). Graphs can be constructed separately from the logs collected by different clients.
    }
    \label{fig:optc}
    \vspace{-3mm}
\end{figure}

On top of the graph data, GNIDS can perform anomaly detection with a trained model.
The relevant works can be divided by their classification targets: sub-graphs (e.g., a graph snapshot), nodes (e.g., a host) and edges (e.g., interactions between hosts). 
In order to accurately model the patterns in the graph data, most works choose a Graph Neural Network (GNN). One prominent technique is \textit{graph autoencoder (GAE)}~\cite{kipf2016variational}, which uses a \textit{graph encoder} to generate node embedding and a \textit{graph decoder} to reconstruct a similar graph. The downstream tasks like edge classification can be done by generating edge scores from node embedding and comparing them with a threshold. In Section~\ref{sec:design}, we elaborate the common design choices of GNIDS models when describing the workflow of \system.

Noticeably, a relevant line of research is provenance- or host-based intrusion detection system (PIDS or HIDS)~\cite{zipperle2022provenance, inam2023sok}, which detects intrusions on the \textit{host log graph}. In Section~\ref{sec:related}, we discuss this line and the potential changes to our system \system for adaptation.

\subsection{Federated Learning}
\label{subsec:fl}

Federated Learning (FL) is an emerging technique to allow multiple clients to train a model without revealing their private data~\cite{konevcny2016federated, mcmahan2017communication}. FL relies on a central parameter server to train a global model under multiple iterations. At the start of each iteration $i$, the server transmits a global model  ($w_i$) to a set of clients ($1, \dots, K$) and they train local models ($w_i^1, \dots, w_i^K$) from $w_i$. Then the clients transmit the local models to the server to average their differences (e.g., FedAVG on model parameters~\cite{mcmahan2017communication}) and generate a new $w_{i+1}$. 

FL has two main deployment settings: \textit{horizontal FL} and \textit{vertical FL}~\cite{yang2019federated}. For horizontal FL, the clients' datasets have a large overlap in the same feature space but little overlap in the sample space  (e.g., every client owns the same type of network logs of its sub-network machines). On the contrary, vertical FL assumes the clients have a large overlap in the sample space but little overlap in the feature space (e.g., each client collects a unique type of logs for all organizational machines). In this work, we focus on horizontal FL, which has been studied more often~\cite{yang2019federated}. 
We also focus on the \textit{cross-silo} FL setting, in which a small number of \textit{organizations} participate in the \textit{entire} training process~\cite{huang2022cross}, rather than the cross-device FL setting, in which many user-owned devices selectively participate in different training iterations.

\zztitle{Federated Graph Learning (FGL).}
Initially, FL was developed for tasks related to Euclidean data like image classification~\cite{mcmahan2017communication}. Recently, FL has been applied to non-Euclidean data like graphs~\cite{liu2022federated, fu2022federated}, and these works are termed under \textit{Federated Graph Learning (FGL)}. Under the horizontal FL setting, the graph data is partitioned across clients, where each client has a sub-graph with non-overlapping (or little overlapping) nodes. 
A prominent challenge for sub-graph FL is the heterogeneity between clients' subgraphs, such that the sizes and topology are vastly different between subgraphs. While there are general solutions to address the data heterogeneity issues under FL~\cite{fedopt,fedprox}, some solutions are customized to sub-graph FL~\cite{li2023fedgta,wang2024fedsg}. %
For example, 
FedGTA proposed a topology-aware optimization strategy for FGL~\cite{li2023fedgta}, but it requires heavy changes on the design of existing graphical models. 

Alternatively, some works propose to \textit{amend} each subgraph with some information shared by other clients or server~\cite{zhang2021fastgnn, chen2021fedgl, zhang2021subgraph, peng2022fedni, zhu2024fedtad}. For example, the server in FedGL asks the clients to upload node embeddings to generate a global pseudo embedding~\cite{chen2021fedgl}. FedSage+ asks the clients to train a neighborhood generator jointly~\cite{zhang2021subgraph}. However, there is no guarantee that the shared information will not leak more clients' information (e.g., node embedding can lead to inference attacks~\cite{zhang2022inference}). 
To mitigate the privacy concerns, cryptographic primitives and/or differential privacy have been tested to amend the subgraph in a privacy-preserving way~\cite{yao2022fedgcn, zhang2021subgraph, qiu2022privacy,  wu2021fedgnn}. For example, FedGCN allows a client to collect 1-hop or 2-hop averaged neighbor node features from clients with Fully Homomorphic Encryption (FHE)~\cite{yao2022fedgcn}. However, significant computation and communication overhead will be incurred. 
In this work, we develop a new FGL technique to tackle the subgraph heterogeneity issue and apply FL to GNIDS in practice.

\section{Overview}
\label{sec:overview}

In this section, we first describe the deployment settings of our system \system. Then, we demonstrate the goals to be achieved by \system.
Finally, we describe the threat model.

\subsection{Deployment Settings}
\label{subsec:problem}

We assume an organization consists of multiple sub-organizations, but it is not always feasible for them to share the raw logs with each other, e.g., when they are located in disjoint regions that are governed by privacy laws like GDPR, as elaborated in Section~\ref{sec:intro}.
Each sub-organization collects the logs about the network packets sent to and received by its controlled machines, with systems like SIEM~\cite{siem}, and trains a local GNIDS model to detect past or ongoing attacks by analyzing the logs.
To achieve better detection coverage and effectiveness, they decide to perform FL to jointly train a global GNIDS model that can be used by each participating sub-organization.
The same procedure can be taken by multiple independent organizations to train a global GNIDS model.

Here we formally define the entities that deploy our system. Figure~\ref{fig:optc} also illustrates the setup.

\begin{itemize}[noitemsep,nolistsep]
    \item A \textbf{client} is the sub-organization that collects logs from its managed machines and trains a GNIDS model to perform intrusion detection.
    \item A \textbf{parameter server} is operated by an entity outside the clients (e.g., the parent organization of the clients) to aggregate the clients' model updates and push the global model to clients.
    \item A \textbf{machine} owned by a client is subjected to attacks. It produces network logs that are collected by the client. Each machine is also called a \textbf{node} under the client graph.
\end{itemize}

\subsection{Design Goals and Challenges}
\label{subsec:goals}

When designing \system, we identify several goals that should be achieved to enable its real-world deployment.

\begin{itemize}[noitemsep,nolistsep]
    \item \textbf{Effectiveness.} \system should achieve high detection accuracy and precision on large-scale real-world logs. Achieving high precision is more important due to the imbalanced data distribution in the log dataset~\cite{quiring2022and}. \system should achieve comparable effectiveness as the GNIDS trained on the entire log dataset.
    \item \textbf{Scalability.} The introduced FL mechanisms should be scalable when training a global model from large-scale log datasets owned by many clients. The communication overhead and latency added by \system on each client should be small.
    \item \textbf{Robustness.} In addition to compromising the client machines, the attacker has the motivation to compromise the FL procedure. \system should be able to defend against such adaptive attacks.
\end{itemize}

\zztitle{Challenges and solutions.}
The major challenge is the heterogeneity among clients. Previous studies have discovered that when clients' data are non-IID (independent and identically distributed), the effectiveness and robustness of the trained global model can be significantly degraded~\cite{li2019convergence}. In our case, it is very likely that each client has divergent subgraphs in terms of size and topology. As supporting evidence, Dambra et al. studied the malware encounters using the telemetry data from a security company, and it shows enterprises in the United States have more than 5x monitored end-hosts than any other country~\cite{dambra2023comparison}.

We found that when clients' data are non-IID graphs, none of the prior FL (e.g., FedAvg) or FGL (e.g., FedGCN) approaches can achieve the aforementioned goals altogether, as they have to \textit{sacrifice one goal for another}.
For example, FedAvg is highly scalable but performs poorly under non-IID data~\cite{li2019convergence}.
FedGCN ~\cite{yao2022fedgcn} aims to address the effectiveness challenge on the heterogeneous client data with heavyweight cryptographic methods, which sacrifices the system scalability.
Norm bounding~\cite{sun2019can} improves the robustness by clipping abnormal model updates, but prior research also demonstrated it will slow down the convergence of FL training and decrease the effectiveness of the trained model when a static clipping threshold is used~\cite{icml/ZhangCH0Y22}.

Hence, new FGL methods are desired in our setting, and our key observation is that by sharing a small piece of information from clients to the parameter server, i.e., the aggregated amount of nodes, the parameter server can adjust the clients' contributions automatically to offset the impact of non-IID graphs. Besides, a client can perform self-adjustment of its contribution based on the divergence of client-to-global model parameters. Since the contributions are dynamically adjusted, the limitation of defenses like norm bounding can be remedied. Overall, through new FGL protocols that are carefully designed around non-IID client graphs, our system \system achieves the three design goals \textit{all together} for the first time.

\subsection{Threat Model}
\label{subsec:threat}

First, we follow the threat model of the other GNIDS (e.g., ~\cite{king2023euler, khoury2024jbeil}) that though the machines can be compromised, the network communications are correctly logged by the network appliances. Hence, log integrity can be achieved. Though it is possible that advanced attackers could violate this assumption, additional defenses (e.g., using Trusted Execution Environment) can be deployed as a countermeasure~\cite{bates2015trustworthy, paccagnella2020custos, gandhi2023rethinking}.

Second, we assume the central parameter server is honest-but-curious, which is trustworthy for aggregation but may be curious about clients' local data (e.g., communications between two employees of a sub-organization).
To mitigate privacy leakage, we only allow the parameter server to know the total number of nodes aggregated from all clients, in addition to the clients' model updates required by FL.
\revisednew{We argue that the total node number is often accessible within an organization: %
e.g., an administrator can get the list of users across sub-organizations from the central active directory server~\cite{get-aduser}), so no extra privacy leakage is expected.} %
\revisednew{In Section~\ref{sec:discussion}, we discuss the privacy threats potentially caused by deploying \system,
including membership inference attack and gradient inversion attack,}
and provide preliminary analysis under the lens of differential privacy (DP).

Finally, we consider the attacker who is able to launch the poisoning attack by controlling one or more FL clients.
The controlled clients are subject to data poisoning~\cite{bagdasaryan2020backdoor} (e.g., the adversary commands some compromised machines to initiate covering communications during the training stage) or model poisoning~\cite{bagdasaryan2020backdoor} (e.g., the adversary manipulates the updates of local models). In Section~\ref{subsec:poison}, we introduce a concrete attack against GNIDS in the FL setting, and demonstrate \system is an effective defense.

\section{Design of \system}
\label{sec:design}

\begin{figure*}[t]
    \centering
    \includegraphics[width=1\textwidth,page=6]{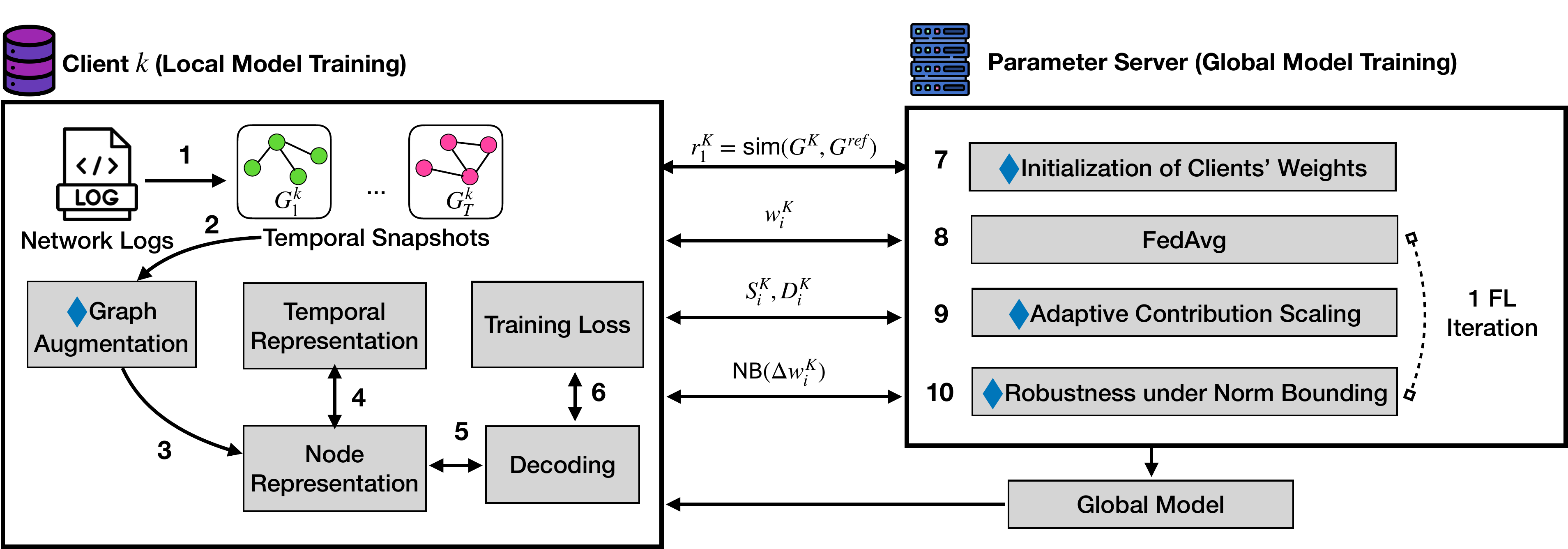}
    \caption{The workflow of \system. The client $k$ trains a GNIDS model locally and communicates with the  parameter server to jointly learn a global GNIDS with other clients. Function $\text{Sim}$ computes graph similarity. 
    }
    \label{fig:workflow}
\end{figure*}

In this section, we describe \system in detail. \system encompasses GNIDS components that are adapted from the existing works and FL components that train the GNIDS models. We highlight FL-related components with ``\novel{}'' as they are the key contributions of this work. The high-level workflow of \system is illustrated in Figure~\ref{fig:workflow}. 
The main symbols and their meanings are summarized in Table~\ref{tab:symbol_table}.

\begin{table}[h!]
    \centering
    \caption{The main symbols used in the paper.}    \setlength\tabcolsep{15pt}
    \begin{tabular}{l|c}
        \hline
        Term(s)  & Symbol(s) \\\hline
        Client number, index & $K$, $k$ \\       
        Client graph & $\mathcal{G}^k$ \\
        Client edges, edge & $\mathcal{E}^k$, $e^k$ \\
        Client nodes, node  & $\mathcal{V}^k$, $v^k$ \\
        Snapshot number, index & $T$, $t$ \\
        Client snapshot & $\mathcal{G}^k_t$ \\
        Client model parameters & $w^k$ \\
        FL max and current iteration & $R$, $i$ \\     
        Weight of a client update & $r^k$ \\  
        Global model after $i$ iterations & $w_{i+1}$\\
        \hline
    \end{tabular}

    \label{tab:symbol_table}
\end{table}

\subsection{Local Graph Creation} 
\label{subsec:local_graph}

We assume $K$ clients have deployed the same GNIDS and they jointly train a global model with FL. Each client is indexed by $k \in [1, K]$. After the network logs are collected by client $k$, a graph $\mathcal{G}^k$ will be constructed by representing the event sources and destinations (machines/users/etc.) as nodes $\mathcal{V}^k$ and their communications as edges $\mathcal{E}^k$. An edge $e^k$ between a pair of nodes $v^k_1$ and $v^k_2$ could contain features extracted from one or many events that have both $v^k_1$ and $v^k_2$, and the commonly used features include the event frequency~\cite{king2023euler}, traffic volume of network flows~\cite{xu2024understanding}, etc. Each node has a feature vector, which can be the node type (e.g., workstation or server), privilege, etc.~\cite{king2023euler}

Though it is relatively straightforward to generate a single static graph from all events~\cite{zeng2022shadewatcher}, such graph modeling has prominent issues like the coarse detection granularity and missing the unique temporal patterns~\cite{xu2024understanding}. Recent works advocate dynamic graph modeling that generates a sequence of \textit{temporal snapshots} $[\mathcal{G}_1, \cdots,$ $\mathcal{G}_T]$, and a snapshot $\mathcal{G}_t$ ($t \in [1, T]$) merges the events in a fix-duration time window (e.g., one hour)~\cite{xu2024understanding, king2023euler, khoury2024jbeil}. \system generates a dynamic graph by default, but it can easily be switched to static modeling by merging the nodes and edges of $[\mathcal{G}^k_1, \cdots, \mathcal{G}^k_T]$, yielding $\mathcal{G}^k$.

\zztitle{\novel{}Augmenting local graph.}
Missing cross-client edges is a prominent issue for subgraph FL, and previous works tried to amend the subgraph by exchanging the topological information among clients, which however led to privacy and efficiency issues, as surveyed in Section~\ref{subsec:fl}.
We found that this problem can be partially addressed under the unique GNIDS deployment setting, with \textit{1-hop graph augmentation}. 
Our key insight is that the log collectors deployed by a client, like firewalls, usually record the inbound and outbound cross-client events  \textit{automatically}~\cite{li2016operational}. 
Therefore, the cross-client edges can be harvested ``for free'' from each client's internal logs.

Specifically, assume $\mathcal{V}_t^k$ and $\mathcal{E}_t^k$ are nodes and edges of the snapshot $\mathcal{G}_t^k$ of client $k$.
We search the logs to find all entities that are not contained by $\mathcal{V}_t^k$ while following the constraints of $\mathcal{V}_t^k$ (e.g., in the same snapshot period), denoted as $\mathcal{V}_t^{k\complement}$. Then $\mathcal{V}_t^{k\complement}$ will be added into $\mathcal{G}_t^k$, together with $\mathcal{E}_t^{k\complement}$ (the edges between $\mathcal{V}_t^k$ and $\mathcal{V}_t^{k\complement}$).

\subsection{Local GNIDS Training}
\label{subsec:local_training}

\zztitle{Node representation learning.}
On a generated graph snapshot $\mathcal{G}^k_t$, \system extracts the representation (or embedding) of each node with a  \textit{graph auto-encoder (GAE)}~\cite{kipf2016variational}, which aggregates the node's features and its neighborhood information. 
Graph Convolutional Network (GCN)~\cite{kipf2016semi} is a standard choice~\cite{king2023euler}, which takes the whole adjacency matrix as the input, but only transductive learning, which assumes the nodes are identical between training and testing, is supported. In our setting, the GCN encoder can be written as:
\begin{equation}\label{eq:gcn}
    Z^k_t=\text{GCN}(X^k_t, A^k_t)
\end{equation}
where $Z^k_t \in \mathbb{R}^{n \times r}$ is the node embedding generated by the encoder on $\mathcal{G}^k_t$, $X^k_t$ represents the node features and $A^k_t$ represents the adjacency matrix on $\mathcal{G}^k_t$.

To accommodate the new nodes observed in the testing phase, inductive learning has been proposed, which learns the neighborhood aggregator to generate node embeddings. GraphSage~\cite{hamilton2017inductive} is a classic encoder in this case, but recent GNIDS~\cite{khoury2024jbeil} also integrates other encoders like Temporal Graph Networks (TGN)~\cite{rossi2020temporal}. The TGN encoder deals with a continuous-time dynamic graph, and for each time $t$, the node embedding $Z^k_t$ is generated by a trainable message function, message aggregator and memory updater.

\zztitle{Temporal learning.}
The anomalous temporal patterns (e.g., a login attempt outside of work hours) provide important indicators for intrusion detection, and some recent GNIDS~\cite{king2023euler, xu2024understanding} choose to capture such patterns with Recurrent Neural Network (RNN), e.g., Gated Recurrent Unit (GRU) under RNN families.
The RNN module can be written as:
\begin{equation}\label{eq:encoder}
    [Z^k_1, \dots, Z^k_T] = \text{RNN}([\text{ENC}(\mathcal{G}^k_1), \dots,  \text{ENC}(\mathcal{G}^k_T)])
\end{equation}
where $\mathcal{G}^k_t = (X^k_t, A^k_t)$ ($t \in [1, T]$)
, $\text{ENC}$ is the encoder (e.g., GCN), and $[Z^k_1, \dots, Z^k_T]$ are the embeddings updated by RNN after they are encoded from  $[\mathcal{G}^k_1, \dots, \mathcal{G}^k_T]$.

When TGN is the encoder, the temporal pattern is directly handled by its memory module, which can be a vanilla RNN and GRU~\cite{rossi2020temporal}. Hence, no extra RNN layer is needed.

Both GCN+RNN and TGN encoders (and other encoders used by GNIDS) are supported by \system when they are trained under FL and we explain this feature in Section~\ref{subsec:proposed_fl}.

\zztitle{Decoding.}
On the node representations, the decoder aims to reconstruct the adjacency matrix with edge scores.
The basic decoder is inner product decoder~\cite{king2023euler}, which can be written as:
\begin{equation}\label{eq:decoder}
    \text{DEC}(Z^k_t) = \Pr(A^k_{t+n}=1 | Z^k_t) = \sigma(Z^k_t Z_t^{k\intercal})
\end{equation}
where $\sigma(\cdot)$ is the logistic sigmoid function and $\Pr(A_{t+n}=1 | Z_t)$ is the reconstructed adjacency matrix at time $t+n$  given $Z_t$. The probability at each matrix cell is used as an edge score.

The simplicity of the inner product could lead to prominent reconstruction loss, and more complex, trainable decoders like Multilayer Perceptron (MLP) have been used by GNIDS~\cite{xu2024understanding,khoury2024jbeil}. We design \system to support different variations of decoders as elaborated in Section~\ref{subsec:proposed_fl}.

\zztitle{Training loss.}
When training a GNIDS, the weights of trainable components, including the graph encoder, RNN and decoder, are updated under a loss function. The typical choice is cross-entropy loss, which can be written as: %
\begin{equation}\label{eq:dec_loss}
    \mathcal{L} = -\log(\Pr(A^k_t|\hat{A^k_t})) 
\end{equation}
where $\hat{A^k_t}$ is the adjacency matrix decoded from node embeddings.

When the GNIDS is trained with only benign events in unsupervised learning mode, 
negative sampling~\cite{mikolov2013distributed} can be applied to randomly select non-existent data points as the malicious samples (e.g., non-existent edges~\cite{king2023euler,xu2024understanding}).

\subsection{Federated GNIDS Training}
\label{subsec:proposed_fl}

We aim to support different GNIDS graph modeling, downstream tasks, training/testing setup, as described in Section~\ref{subsec:goals}. Hence, \system is designed to augment the existing GNIDS without heavy adjustment of their components, and the FGL works that redesign the local models~\cite{meng2021cross,li2023fedgta} are not suitable. Based on our survey of existing FL frameworks, we are motivated to build \system on top of FedAVG, because 1) it only needs clients' model parameters as input, and 2) it has demonstrated effectiveness when the clients' models are GNN~\cite{wang2022federatedscope} and RNN~\cite{hard2018federated}. Since the common GNIDS components like node representation learning, temporal learning and decoding all use the same set of training data, we could train them independently with FedAvg and update their model parameters.
As such, \system not only incurs minimum development overhead, but also achieves similar or even better effectiveness compared to the original GNIDS, as supported by the evaluation (Section~\ref{sec:evaluation}). We also want to point out that the standard FedAvg workflow leads to unsatisfactory performance and we elaborate on the key adjustment below.

\zztitle{\novel{}Initialization of clients' weights.}
FedAVG aggregates the model parameters of GNIDS components in a number of iterations as described in Section~\ref{subsec:fl}. Its basic version in one iteration can be represented as:
\begin{equation}\label{eq:fedavg}
    w_{i+1} = \Sigma^{K}_{k=1} r^k \times w^k_i
\end{equation}
where $i$ is the current iteration, $w^k_i$ is the local model parameters of client $C_k$ that is trained with the global model of previous iteration $w_i$ and its local data, $r^k$ is the weight of client $k$ during aggregation, and $w_{i+1}$ is the global model parameters after aggregation. As described in Section~\ref{subsec:goals}, the clients' data are non-IID, which has a prominent impact on the performance of FL algorithms.  
To tackle non-IID data, one feasible approach is to assign different weights to the model updates from different clients, so the impact from clients with outlying distributions can be contained. Previous works have used the number of data samples (e.g., images~\cite{mcmahan2017communication}) per client for weights, but in our case, each client only has \textit{one} sub-graph. 

We address this challenge with a new method to initialize the client's weights based on its \textit{graph properties}.
We minimize the communication overhead and privacy leakage by computing the client's weights on top of \textit{a single reference graph} generated by the parameter server. 
This reference graph stays the same for the whole training process and across clients.
Our approach is inspired by Zhao et al., which distributes a warm-up model trained with globally shared data and shows test accuracy can be increased~\cite{zhao2018federated}.

Specifically, we assume the server knows the total number of nodes ($n$) from all clients, \revisednew{as justified in Section~\ref{subsec:threat}}. When bootstrapping FL, the parameter server generates a reference graph $\mathcal{G}^{ref}$ and distributes it to all clients for them to compute $r^k$ ($\forall k \in [1, K]$). We choose to apply \textit{Barabási–Albert (BA) Model}~\cite{barabasi1999emergence} to generate the reference graph. 
BA model is a \textit{random graph model} for complex networks analysis~\cite{drobyshevskiy2019random}, and it is selected

because 1) it only needs the number of nodes ($n$) and the number of initial edges for a new node ($m$) and 2) it has low computational overhead (complexity is $O(n \times m)$) when $m$ is small. 
In Section~\ref{subsec:settings} and Appendix~\ref{app:ablation}, we discuss the selection of $m$ and evaluate its impact. The pseudo-code of BA model is written in Appendix~\ref{app:ba}.

\noindent\textbf{\novel{}Graph sketching.}
On a client $C^k$, $r^k$ will be initialized based on the \textit{graph similarity} between its generated graph $\mathcal{G}^k$ (aggregated from $[\mathcal{G}_1^k, \dots, \mathcal{G}_T^k]$) and $\mathcal{G}^{ref}$. Intuitively, the client with a similar distribution to the global data should receive high $r^k$. Noticeably, we compute $r^k$ locally on the client, so nothing about the client's data or distribution will be learned by the server. This is different from the standard FedAVG, which computes weights on the server. However, graph similarity is computationally intensive: e.g., graph edit distance (GED) is a fundamental NP-hard problem~\cite{bunke1997relation}. Therefore, we choose to compute the similarity on the \textit{graph sketches} instead of raw graphs for efficiency. 
We use \textit{Weisfeiler-Lehman (WL) graph kernel}~\cite{shervashidze2011weisfeiler} to capture the graph structure surrounding each node and compute \textit{histogram} by nodes' neighborhood.

\begin{algorithm}[h]
\caption{Weisfeiler-Lehman Histogram (WLH). %
}
\label{alg:wlh}
\KwData{Graph \( \mathcal{G} \)}
\KwResult{histograms }
\BlankLine
$\text{labels} \gets \text{InitializeLabels}(\mathcal{G})$\;
$\text{histograms} \gets \text{Histogram}(\text{labels})$\;
\For{\( i = 1 \) \KwTo MaxIters}{
    \ForEach{\( v \) in \( \mathcal{G} \)}{
        $\text{neigh} \gets \text{NeighborLabels}(v, \text{labels})$\;
        $\text{labels}[v] \gets \text{Hash}(\text{Sort}(\text{labels}[v] + \text{neigh}))$\;
    }
    \( \text{histograms} \gets \text{Append}(\text{histograms}, \text{Histogram}(\text{labels})) \)\\
}
\Return \( \text{histograms} \)
\end{algorithm}
In Algorithm~\ref{alg:wlh}, we describe how to compute the Weisfeiler-Lehman histogram (WLH), where $\text{InitializeLabels}$ assigns \textit{node degree} as a label to each node $v$, $\text{NeighborLabels}$ creates a multiset containing $v$'s current label and its neighbors' labels, $\text{Hash}$ and $\text{Sort}$ produce a new label for $v$.
$MaxIters$ is the number of iterations, and $i$th-iteration computes WLH for $i$-hop neighborhood. We set $MaxIters$ to 3.
When applying Algorithm~\ref{alg:wlh} to our setting, we compute the Jaccard similarity $\jsequ^k$ between $\mathcal{G}^k$ and $\mathcal{G}^{ref}$. 
\begin{equation}
\label{equ:rk}
   \jsequ^k =  \frac{\mathcal{G}^{ref}\cap \mathcal{G}^k}{\mathcal{G}^{ref}\cup \mathcal{G}^k}
\end{equation}

\noindent\textbf{\novel{}Adaptive contribution scaling (ACS).}
We follow the FedAVG client-server protocol to update the model parameters by iterations. While FedAvg uses the same $r^k$ throughout FL, we found the client contributions can be more precisely modeled by dynamically adjusting them towards stability.
This idea at the high level has been examined by prior works like ~\cite{siomos2023contribution, wu2021fast, cao2021fltrust}, but we found none of them are suitable under our setting. First, they require the statistical information of local data (e.g., local gradients or local samples~\cite{wu2021fast,cao2021fltrust}), which has been avoided by \system due to privacy concerns. Second, the contribution evaluation can be heavy (e.g., Shapley value by ~\cite{siomos2023contribution}).

Instead, we propose a \textit{lightweight and privacy-preserving} method termed \textit{adaptive contribution scaling (ACS)} to adjust clients' weights only using the model parameters.
Specifically, for client $C^k$ at the FL iteration $i$, the parameter server computes a similarity metric
$S^k_i$ based on the cosine similarity, and a distance metric based on L2 distance $D^k_i$. The new client' weights $r^k_i$ will be computed with $\jsequ^k$, $S^k_i$ and $D^k_i$.
$r^k_1$ is initialized using $\jsequ^k$. ACS models both representational similarity (with $S^k_i$) and distance of local models to the global models (with $D^k_i$), hence presenting a more reliable indicator for client contributions.
To mitigate model collapse incurred by some clients, the server bounds $D^k_i$ by a hyperparameter $\omega$.

The definitions of $S^k_i$ and $D^k_i$, and the new aggregation function are listed below. 
\begin{equation}\label{eq_sim}
    S^k_i = \frac{||\sum w^k_i \times w_{i-1}||}{||\sum w_i^k||\times ||w_{i-1}||}
\end{equation}
\begin{equation}\label{eq_dis}
    D^k_i = \frac{\omega\sqrt{\sum(w_i^k-w_{i-1})^2}}{\max(\omega, \sqrt{\sum(w_i^k-w_{i-1})^2})},\forall\omega>0
\end{equation}

\begin{equation}
\label{equ:w_update}
    w_{i+1} = \frac{1}{K}\sum^{K}_{k=1} r^k_{i} \times w^k_i \quad \mathrm{s.t.} \ r^k_{i} =  c_1 \times \jsequ^k + c_2 \times S^k_i \times D^k_i
\end{equation} 
where $c_1$ and $c_2$ are two constants to adjust the contributions from $\jsequ^k$, $S^k_i$ and $D^k_i$.

With ACS, \system is able to achieve better effectiveness in the evaluation. Moreover, we are able to \textit{formally prove} that the iteration-wise difference shifting under Equation~\ref{equ:w_update} is \textit{bounded}, by $|\frac{w_{i+1}}{\sum_k w^k_i}|\leq c_1+c_2\omega $ from $\sum_k w^k_i$ (FedAVG). 
The non-IID data can lead to slow convergence or even non-convergence in training~\cite{li2019convergence}. By bounding the iteration-wise difference shifting through  ACS (i.e., adjusting the e contributions from $\jsequ^k$, $S^k_i$ and $D^k_i$), we are able to address this issue with a formal guarantee.

\begin{theorem}
\label{the:diff_from_fedavg}
    Define $\jsequ^k, S^k_i$, and $D_i^k$ as Equation~\ref{equ:rk}, Equation~\ref{eq_sim}, and Equation~\ref{eq_dis}, respectively.
    Let $c_1,c_2$ be two hyperparameters to adjust the contributions.
    Global model update shifting from FedAVG per iteration $i$ is bounded by $|\frac{w_{i+1}}{\sum_k w^k_i}|\leq c_1+c_2\omega $ for any $0\leq \jsequ^k,S^k_i\leq1$.
\end{theorem}

In Appendix~\ref{app:diff_from_fedavg}, we show the full proof of Theorem~\ref{the:diff_from_fedavg}. This analytical result shows promises that \system can be effective in other applications that use subgraph FL.

\noindent\textbf{\novel{}Robustness via norm bounding.}
Due to the use of FL, a client deploying GNIDS could be more vulnerable, when the other clients are compromised and poisoning the global model. We observe that ACS can be extended to defend against FL poisoning attack by limiting the contribution of a client, and we are motivated to integrate 
norm bounding (NB)~\cite{sun2019can} for this purpose. In particular, NB observes that the attacker's model updates are likely to have large norms to influence the direction of the global model. 
As a countermeasure, the server can bound the model updates with a threshold $M$ to mitigate the impact of the abnormal update during aggregation. 
Our approach differs from the standard NB by adjusting the bounds \textit{dynamically} under ACS.

We define $\mathsf{NB}(\Delta w^{k}_{i+1})=\frac{\Delta w^{k}_{i+1}}{\max(1, \lVert \Delta w^{k}_{i+1} \rVert_2 / M)}$~\cite{sun2019can} for simplification.  The global model is bounded through the updated difference $\Delta w^{k}_{i+1}$ derived from local models as Equation~\ref{eq_update_bounding}. 
\begin{equation}
\label{eq_update_bounding}
    w_{i+1}=w_i+\frac{1}{K}\sum^{K}_{k=1}r_i^k\times  \mathsf{NB}(\Delta w^{k}_{i+1})\quad s.t. \ \Delta w^{k}_{i+1}=w^{k}_{i+1}-w_i
\end{equation}

Like ACS, we present Theorem~\ref{the:convergence} (formally proved in Appendix~\ref{app:nb_proof}) to confirm that \system has a bounded convergence rate for establishing a global model. We want to highlight that NB could lead to slow convergence and worse effectiveness~\cite{icml/ZhangCH0Y22} in the sacrifice of robustness, but such an issue is well resolved when combining ACS and NB, achieving \textit{dynamic clipping}. In Appendix~\ref{app:dpfl}, we further extend Equation~\ref{eq_update_bounding} to support differential privacy (DP), when the privacy of model output is a concern. Overall, our method entails a strong guarantee that effectiveness and robustness can be preserved concurrently on non-IID graph data.

\begin{theorem}
\label{the:convergence}
Assume all Lemmas~\ref{ass:lip_bound},\ref{lem:bvar},\ref{lem:bvarg},\ref{lem:bgra} and constrains~\cite{iclr/ReddiCZGRKKM21} reviewed in Appendix~\ref{app:proof}.
  For any bounding norm  $M\geq \eta EG$, theoretical complexity of \system's convergence  is $O(1/{(R\eta E(c_1+c_2\omega))})+O((c_1+c_2\omega)\eta/K\cdot\sigma^2_{\local})+O(\eta^2E^2\cdot\sigma^2_{\globalsf})$.
\end{theorem}

\subsection{Intrusion Detection}
\label{subsec:anomaly}

After the model is trained, for intrusion detection by GNIDS, different classification granularities can be applied. Edge classification compares the edge scores generated from the GNIDS decoder with a threshold, which can be learnt from validation snapshots, and achieves the finest detection granularity~\cite{king2023euler,xu2024understanding,khoury2024jbeil}. Node classification adds a classification layer (e.g., softmax) on top of the node embeddings and compares the probabilities with a threshold~\cite{khoury2024jbeil}. Alternatively, one can compute a score for the whole snapshot by aggregating the edge scores and detect the abnormal ones. In this work, we test \system over edge-level classification due to its finest detection granularity. Supporting other detection granularities is trivial by applying the aforementioned changes.

In Appendix~\ref{app:workflow}, we also summarize the detailed workflow of \system in pseudo-code.

\section{Evaluation}
\label{sec:evaluation}

This section describes our experiment setting, including the evaluated GNIDS, datasets, baseline FL methods, etc. Then, we consider the effectiveness of \system and other baseline FL methods on different combinations of GNIDS and datasets. Next, we show that \system fulfills the scalability and robustness goals. An ablation study is performed in the end to understand the impact of individual components and hyper-parameters.

\subsection{Experiment Settings}
\label{subsec:settings}

\zztitle{Evaluated GNIDS.}
We adapt two GNIDS models Euler~\cite{king2023euler} and Jbeil~\cite{khoury2024jbeil} under \system. 
We chose them because they have open-source implementations and both have been tested on large-scale network datasets. In addition, their architectures are very different, so we can assess whether the design goals can be achieved across different GNIDS modes.
For instance, due to the use of GCN, Euler only supports transductive learning, while TGN used by Jbeil supports both transductive learning and inductive learning. Both learning modes are tested by us.

\zztitle{Datasets.} We use OpTC, LANL cyber1 (or LANL for short) and Pivoting datasets to evaluate \system. These datasets have been used by our baseline GNIDS~\cite{king2023euler,khoury2024jbeil} and thoroughly tested by other GNIDS like~\cite{xu2024understanding,paudel2022pikachu}. 
In Appendix~\ref{app:datasets}, we describe the characteristics of the datasets and how we pre-process them.
Table~\ref{tab:dataset_all} shows their statistics. 

\revisednew{All datasets are highly non-IID as reflected in Table~\ref{tab:kl_client} of Appendix. For LANL and OPTC, we observe high standard deviation on node number in different clients and in Pivoting, the event numbers of clients have high standard deviation. Such dataset characteristic justifies the design of \system.
}

\begin{table}[t]
    \centering
    \caption{The statistics of the three tested datasets. %
    }    
    \begin{tabular}{l|cccc}
        \hline
        Dataset & \#Nodes & \#Events & Duration  \\
        \hline
        OpTC & 814 & 92,073,717  & 8 days &  \\ 
        OpTC-redteam & 28 & 21,784  & 3 days \\
        \hline
        LANL & 17,649 & 1,051,430,459 & 58 days &  \\ 
        LANL-redteam & 305  & 749  & 28 days \\
        \hline
        Pivoting &  1,015 & 74,551,643 & 1 day\\ 
        \hline
        
    \end{tabular}
    \label{tab:dataset_all}
\end{table}

\zztitle{Data split for clients.} 
To simulate the FL process, we need each client to keep a subgraph of the complete graph. The prior FGL studies choose to cluster the nodes and generate local graphs~\cite{yao2022fedgcn, zhang2021subgraph, xie2021federated}. We follow this direction and leverage a recent approach Multilayer Block Model (MBM)~\cite{larroche2022multilayer} to cluster the logs. MBM clusters system events to visualize the major clusters and major events between different clusters. 
For LANL, MBM has built-in support and we use its code to generate a user-machine bipartite graph and change its ``Number of bottom clusters''  to get different numbers of sub-graphs~\cite{mbm}. 
For OpTC, 
we categorize its hosts into internal and external nodes, like how MBM processes the VAST dataset (a dataset with simulated network traffic), and generate sub-graphs like processing LANL.
For Pivoting, we follow similar settings as LANL except that we replace machines with hosts, and use the protocol and destination port to fill the node type.

\zztitle{Metrics.} 
For Euler, we define true positive (TP) as a malicious edge in a snapshot that contains at least one redteam event, and true negative (TN) as a normal edge without any redteam event. False positive (FP) and false negative (FN) are misclassified as malicious and normal edges. 
For Jbeil, we consider TP as a non-existent edge and TN as an existing edge while FPs and FNs are misclassified as non-existent and existing edges
With TP, TN, FP and FN defined, we compute precision, recall and FPR as $\frac{TP}{TP+FP}$, $\frac{TP}{TP+FN}$ and $\frac{FP}{TP+FP}$.

The values of precision, recall and FPR depend on the classification threshold, which might not always be optimal, e.g., when the validation dataset has a different distribution from the testing dataset. Hence, we also compute the area under the ROC curve (AUC) which is computed over all thresholds. When the malicious and normal edges are highly imbalanced, AUC might not correctly capture the effectiveness of a GNIDS, as it measures the relation between TPR and FPR (see ``Base Rate Fallacy'' in ~\cite{quiring2022and}). So, we also compute average precision (AP), which is defined as:
\begin{equation}\label{eq:ap}
\begin{aligned}
\text{AP} = \sum_n  (R_n - R_{n-1}) \times P_n
\end{aligned}
\end{equation}
where $R_n$ and $P_n$ are the precision and recall at the $n$-th threshold. It measures the area under the precision-recall curve, which ``conveys the true performance''~\cite{quiring2022and} on an imbalanced dataset. We use AUC and AP as the major metrics in the percentage format following \cite{ren2015faster,paisitkriangkrai2015pedestrian}.

\zztitle{Baselines.} 
We consider 5 baseline models to compare with \system. Except for Non-FL, all the other models are well-known FL models or extensions. For them, the cross-client edges have been added to the local graphs for a fair comparison with \system. 

\begin{itemize}[noitemsep,nolistsep]
    \item \textbf{Non-FL.} For this method, we assume GNIDS is directly trained on the whole training dataset and FL is not involved. 
    
    \item \textbf{FedAVG~\cite{mcmahan2017communication}.} We use the basic version and use the same weight for each client.
    
    \item \textbf{FedOpt~\cite{fedopt}.}
    FedOpt seeks to improve the convergence and stability of FL in non-IID settings. It employs adaptive local solvers and a server-side momentum term to achieve faster convergence.
    
    \item \textbf{FedProx~\cite{fedprox}.}
    Similar to FedOpt, FedProx is designed to handle non-IID settings. It introduces a proximal term to the local optimization objective of each client, which helps prevent divergence when local datasets are skewed.

    \item \textbf{\systemNumber.} This is a simple extension of FedAVG, in which a client's weight depends on its node numbers $n^k$, i.e., $r^k = \frac{n^k}{n}$, where $n$ is the total node number. $r^k$ is constant across iterations.

\end{itemize}

To evaluate the impact of norm bounding on GNIDS utility, we create a variation of \system, termed \system-UB, that performs without Equation~\ref{equ:w_update}.

\zztitle{Hyper-parameters.} We set $m=5$ for Barabási–Albert (BA) model. 
For the client numbers $K$, we vary it from 2-5 for LANL, 2-5 OpTC, and 2-4 for Pivoting. 
We found if we further increase $K$, MBM will yield invalid clusters (e.g., empty clusters). For Pivoting, because the edge features useful for clustering are fewer (only the port number can be used), generating clusters with sufficient node numbers from $K \geq 5$ is infeasible. Yet, for the scalability analysis in Section~\ref{subsec:scalability}, we tried $K=10,20$ for LANL by dropping invalid clusters.
The number of local epochs $E$ is set to 1, except for FedOpt, for which we set $E$ to 5. For FedProx, we set its parameter $\mu$, which controls the weight of the proximal term, to 0.05.
For the learning rate $\eta$, we set 0.01 for LANL, 0.005 for OpTC and 0.003 for Pivoting. We use Adam Optimizer.
For ACS hyper-parameters, we set a large $c_1$ as 0.8 and a small $c_2$ as 0.2 to avoid drastic changes on $w_{i+1}$. $\omega$ is set as 5.
For norm bounding parameter $M$, we use  5 for Euler and 70 for Jbeil. We use a larger value of $M$ for Jbeil because the norm $w$ of Jbeil is empirically large.

\zztitle{Environment.} 
We run the experiments on a workstation that has an AMD Ryzen Threadripper 3970X 32-core Processor and 256 GB CPU memory. The OS is Ubuntu 20.04.2 LTS. 
We use PyTorch 1.10 and Python 3.9.12 as the environment when building \system. We use GPU implementations as default and our GPU is NVIDIA GeForce RTX 3090 with 24GB memory.
For Euler, we use its source code from~\cite{euler-github} and disable its distributed setting so FL can be implemented. 
For Jbeil, we use its source code from~\cite{jbeil-github}.

\subsection{Effectiveness}
\label{subsec:effectiveness}
We evaluate the overall performance of \system under the default hyper-parameters
and compare the results with the other baselines. We list the results when $K=3-5$ and the results under $K=2$ are shown in the Appendix~\ref{app:lanloptc}.
We assume no poisoning attacks happen here and leave the evaluation under poisoning attacks in Section~\ref{subsec:poison}.

\setlength{\tabcolsep}{4pt}
\begin{table*}[h]
    \scriptsize
    \centering
    \caption{Evaluation on Euler and Jbeil under different client numbers $K$. Non-FL is listed for ease of comparison. All numbers are used in the percent format. The \colorbox{colorbest}{\color{colortext}{best}} and the \colorbox{colorsecond}{\color{colortext}second-best} among FL methods are highlighted respectively. \greenup means outperforming non-FL and \reddown means the worst performance.
    }
    \begin{tabular}{l|cc|cc|cc|cc|cc|cc|cc|cc}
    \hline
    \multicolumn{9}{c|}{\textbf{{\small OPTC-Euler}}} & \multicolumn{8}{c}{\textbf{{\small LANL-Euler}}} \\
    {\footnotesize Client\#} & \multicolumn{2}{c}{\textbf{2}} & \multicolumn{2}{c}{\textbf{3}} & \multicolumn{2}{c}{\textbf{4}}& \multicolumn{2}{c|}{\textbf{5}} & \multicolumn{2}{c}{\textbf{2}} & \multicolumn{2}{c}{\textbf{3}}& \multicolumn{2}{c}{\textbf{4}}& \multicolumn{2}{c}{\textbf{5}} \\
    \hline
    {\footnotesize Algorithm} & AP & AUC & AP & AUC & AP & AUC  & AP & AUC & AP & AUC & AP & AUC & AP & AUC  & AP & AUC \\
    \hline
    {\footnotesize Non-FL}      & 69.96$\tbspace$ & 98.76$\tbspace$ & 69.96$\tbspace$ & 98.76$\tbspace$ & 69.96$\tbspace$ & 98.76$\tbspace$ & 69.96$\tbspace$ & 98.76$\tbspace$  & 0.89$\tbspace$ & 97.93$\tbspace$ & 0.89$\tbspace$ & 97.93$\tbspace$ & 0.89$\tbspace$ & 97.93$\tbspace$ & 0.89$\tbspace$ & 97.93$\tbspace$ \\
    \hline
    {\footnotesize FedAVG} & 33.69\reddown & 95.78\reddown & 65.72\reddown & 98.62\reddown & 66.41\reddown & 97.00\reddown & 75.19\greenup & 98.66$\tbspace$ & 0.71$\tbspace$ & 98.08\greenup & 0.10$\tbspace$ & \cellcolor{colorsecond}\color{colortext}98.08\greenup & 0.22$\tbspace$ & 98.41\greenup & 0.57$\tbspace$ & \cellcolor{colorbest}\color{colortext}{98.79}\greenup \\
    {\footnotesize FedOpt}       & 66.62$\tbspace$ & 98.37$\tbspace$ & 67.24$\tbspace$ & 99.19\greenup & 70.09\greenup & 98.26$\tbspace$ & 72.19\greenup & 98.84\greenup& 0.65$\tbspace$ & 97.87$\tbspace$ & 0.60$\tbspace$ & \cellcolor{colorbest}\color{colortext}{98.80}\greenup & 0.60$\tbspace$ & \cellcolor{colorbest}\color{colortext}{99.08}\greenup & 0.01\reddown & 84.45$\tbspace$ \\
    {\footnotesize FedProx}  & 64.29$\tbspace$ & 98.07$\tbspace$ & 72.46\greenup & 99.25\greenup & 73.79\greenup & 97.30$\tbspace$ & 71.38\greenup & 98.25\reddown & \cellcolor{colorsecond}\color{colortext}0.72$\tbspace$ & 97.90$\tbspace$ & 0.01\reddown & 81.98\reddown & 0.19\reddown & \cellcolor{colorsecond}\color{colortext}98.84\greenup & 0.01\reddown & 81.64\reddown \\
    {\footnotesize \systemNumber} & 70.41\greenup & 99.13\greenup & 69.91$\tbspace$ & 99.29\greenup & 78.40\greenup & 98.97\greenup & 71.87\greenup & 98.78\greenup & 0.51\reddown & 96.85\reddown & 0.62$\tbspace$ & 97.97\greenup & 0.41$\tbspace$ & 96.80\reddown & 0.53$\tbspace$ & 96.68$\tbspace$ \\
    {\footnotesize \system-UB}   & \cellcolor{colorbest}\color{colortext}{77.14}\greenup & \cellcolor{colorbest}\color{colortext}{99.50}\greenup & \cellcolor{colorsecond}\color{colortext}82.95\greenup & \cellcolor{colorsecond}\color{colortext}99.68\greenup & \cellcolor{colorbest}\color{colortext}{84.96}\greenup & \cellcolor{colorsecond}\color{colortext}99.69\greenup & \cellcolor{colorsecond}\color{colortext}78.12\greenup & \cellcolor{colorsecond}\color{colortext}99.24\greenup  & 0.70$\tbspace$ & 97.07$\tbspace$ & \cellcolor{colorsecond}\color{colortext}0.65$\tbspace$ & 97.95\greenup & \cellcolor{colorbest}\color{colortext}{0.82}$\tbspace$ & 97.01$\tbspace$ & \cellcolor{colorsecond}\color{colortext}0.65$\tbspace$ & \cellcolor{colorsecond}\color{colortext}97.33$\tbspace$ \\
    {\footnotesize \system}     & \cellcolor{colorsecond}\color{colortext}74.60\greenup & \cellcolor{colorsecond}\color{colortext}99.41\greenup & \cellcolor{colorbest}\color{colortext}{83.29}\greenup & \cellcolor{colorbest}\color{colortext}{99.76}\greenup & \cellcolor{colorsecond}\color{colortext}80.82\greenup & \cellcolor{colorbest}\color{colortext}{99.79}\greenup & \cellcolor{colorbest}\color{colortext}{82.03}\greenup & \cellcolor{colorbest}\color{colortext}{99.79}\greenup & \cellcolor{colorbest}\color{colortext}{0.77}$\tbspace$ & \cellcolor{colorbest}\color{colortext}{98.92}\greenup & \cellcolor{colorbest}\color{colortext}{0.87}$\tbspace$ & 97.68$\tbspace$ & \cellcolor{colorsecond}\color{colortext}0.72$\tbspace$ & 97.00$\tbspace$ & \cellcolor{colorbest}\color{colortext}{0.67}$\tbspace$ & 97.16$\tbspace$ \\    
    \hline
    \multicolumn{9}{c|}{\textbf{{\small LANL-Jbeil-Transductive}}} & \multicolumn{8}{c}{\textbf{{\small LANL-Jbeil-Inductive}}} \\
    {\small Client\#} & \multicolumn{2}{c}{\textbf{2}} & \multicolumn{2}{c}{\textbf{3}} & \multicolumn{2}{c}{\textbf{4}}& \multicolumn{2}{c|}{\textbf{5}} & \multicolumn{2}{c}{\textbf{2}} & \multicolumn{2}{c}{\textbf{3}}& \multicolumn{2}{c}{\textbf{4}}& \multicolumn{2}{c}{\textbf{5}} \\
    \hline
    {\footnotesize Algorithm} & AP & AUC & AP & AUC & AP & AUC  & AP & AUC & AP & AUC & AP & AUC & AP & AUC  & AP & AUC \\
    \hline
    {\footnotesize Non-FL}       & 96.47$\tbspace$ & 97.51$\tbspace$ & 96.47$\tbspace$ & 97.51$\tbspace$ & 96.47$\tbspace$ & 97.51$\tbspace$ & 96.47$\tbspace$ & 97.51$\tbspace$ & 94.47$\tbspace$ & 96.21$\tbspace$ & 94.47$\tbspace$ & 96.21$\tbspace$ & 94.47$\tbspace$ & 96.21$\tbspace$ & 94.47$\tbspace$ & 96.21$\tbspace$ \\
    \hline
    {\footnotesize FedAVG}       & 88.54$\tbspace$ & 91.85$\tbspace$ & 66.77\reddown & 72.98\reddown & 50.94\reddown & 51.84\reddown & 59.64\reddown & 65.11\reddown & 87.30$\tbspace$ & 91.34$\tbspace$ & 67.25\reddown & 73.89\reddown & 50.06\reddown & 50.12\reddown & 58.77\reddown & 63.86\reddown \\
    {\footnotesize FedOpt}       & 59.67\reddown & 64.64\reddown & 82.18$\tbspace$ & 85.81$\tbspace$ & 60.75$\tbspace$ & 67.35$\tbspace$ & 72.58$\tbspace$ & 80.06$\tbspace$  & 58.79\reddown & 63.50\reddown & 82.00$\tbspace$ & 86.12$\tbspace$ & 62.19$\tbspace$ & 69.20$\tbspace$ & 75.16$\tbspace$ & 82.33$\tbspace$ \\
    {\footnotesize FedProx}      & 61.00$\tbspace$ & 67.71$\tbspace$ & 67.29$\tbspace$ & 74.32$\tbspace$ & 70.89$\tbspace$ & 77.87$\tbspace$ & 76.52$\tbspace$ & 82.02$\tbspace$ & 62.80$\tbspace$ & 69.94$\tbspace$ & 68.28$\tbspace$ & 75.29$\tbspace$ & 70.84$\tbspace$ & 77.96$\tbspace$ & 76.52$\tbspace$ & 81.88$\tbspace$ \\
    {\footnotesize \systemNumber} & 71.37$\tbspace$ & 77.46$\tbspace$ & \cellcolor{colorbest}\color{colortext}{88.96}$\tbspace$ & \cellcolor{colorbest}\color{colortext}{91.14}$\tbspace$ & 61.92$\tbspace$ & 68.93$\tbspace$ & 63.30$\tbspace$ & 70.16$\tbspace$ & 67.33$\tbspace$ & 74.13$\tbspace$ & \cellcolor{colorbest}\color{colortext}{90.50}$\tbspace$ & \cellcolor{colorbest}\color{colortext}{92.58}$\tbspace$ & 63.85$\tbspace$ & 71.31$\tbspace$ & 62.80$\tbspace$ & 69.77$\tbspace$ \\
    {\footnotesize \system-UB}   & \cellcolor{colorbest}\color{colortext}{95.40}$\tbspace$ & \cellcolor{colorbest}\color{colortext}{96.47}$\tbspace$ & 77.47$\tbspace$ & 82.98$\tbspace$ & \cellcolor{colorbest}\color{colortext}{90.29}$\tbspace$ & \cellcolor{colorbest}\color{colortext}{91.47}$\tbspace$ & \cellcolor{colorbest}\color{colortext}{85.46}$\tbspace$ & \cellcolor{colorsecond}\color{colortext}{86.22}$\tbspace$ & \cellcolor{colorbest}\color{colortext}{95.54}\greenup & \cellcolor{colorbest}\color{colortext}{96.55}\greenup & 82.12$\tbspace$ & 86.99$\tbspace$ & \cellcolor{colorbest}\color{colortext}{91.28}$\tbspace$ & \cellcolor{colorbest}\color{colortext}{92.65}$\tbspace$ & \cellcolor{colorbest}\color{colortext}{88.09}$\tbspace$ & \cellcolor{colorbest}\color{colortext}{88.83}$\tbspace$ \\
    {\footnotesize \system}     & \cellcolor{colorsecond}\color{colortext}94.43$\tbspace$ & \cellcolor{colorsecond}\color{colortext}96.25$\tbspace$ & \cellcolor{colorsecond}\color{colortext}87.47$\tbspace$ & \cellcolor{colorsecond}\color{colortext}88.80$\tbspace$ & \cellcolor{colorsecond}\color{colortext}84.10$\tbspace$ & \cellcolor{colorsecond}\color{colortext}85.09$\tbspace$ & \cellcolor{colorsecond}\color{colortext}81.05$\tbspace$ & \cellcolor{colorbest}\color{colortext}87.70$\tbspace$  & \cellcolor{colorsecond}\color{colortext}93.58$\tbspace$ & \cellcolor{colorsecond}\color{colortext}95.88$\tbspace$ & \cellcolor{colorsecond}\color{colortext}89.13$\tbspace$ & \cellcolor{colorsecond}\color{colortext}90.49$\tbspace$ & \cellcolor{colorsecond}\color{colortext}87.80$\tbspace$ & \cellcolor{colorsecond}\color{colortext}86.36$\tbspace$ & \cellcolor{colorsecond}\color{colortext}80.56$\tbspace$ & \cellcolor{colorsecond}\color{colortext}87.48$\tbspace$ \\
    \hline
    \end{tabular}
    \label{tab:whole}
\end{table*}

\zztitle{Results on OpTC.}
In Table~\ref{tab:whole}, we compare the AUC and AP of different FL methods, with different client numbers $K$.
First, we found \system and \system-UB clearly outperform the other FL methods for \textit{every} $K$, reaching 74\%-84\% AP and over 99\% AUC. The accuracy loss due to norm bounding is also acceptable. 
Simply initializing the clients' weights with node numbers (\systemNumber) boosts the performance of FedAvg, but not yet reached the same level of \system (e.g., 69.96\% AP vs 83.29\% AP under $K=3$). 
Though FedOpt and FedProx are designed to address the non-iid issue of FedAVG, their performance is even worse than the simple variation of FedAvg (\systemNumber), suggesting pure data-adaptive tuning is difficult on imbalanced datasets.

Interestingly, we found that \system even outperformed Non-FL, which trains the GNIDS using all data.
In fact, the AP achieved by \system is 5\%-14\% higher compared to Non-FL. 
This might seem counter-intuitive, but a similar observation was also documented in a recent work ~\cite{yan2025you}, which shows that if class imbalance and data heterogeneity are well handled, FL methods can achieve better results than non-FL training on the global long-tailed data.

\begin{figure}[h]
    \centering    \includegraphics[width=0.45\textwidth]{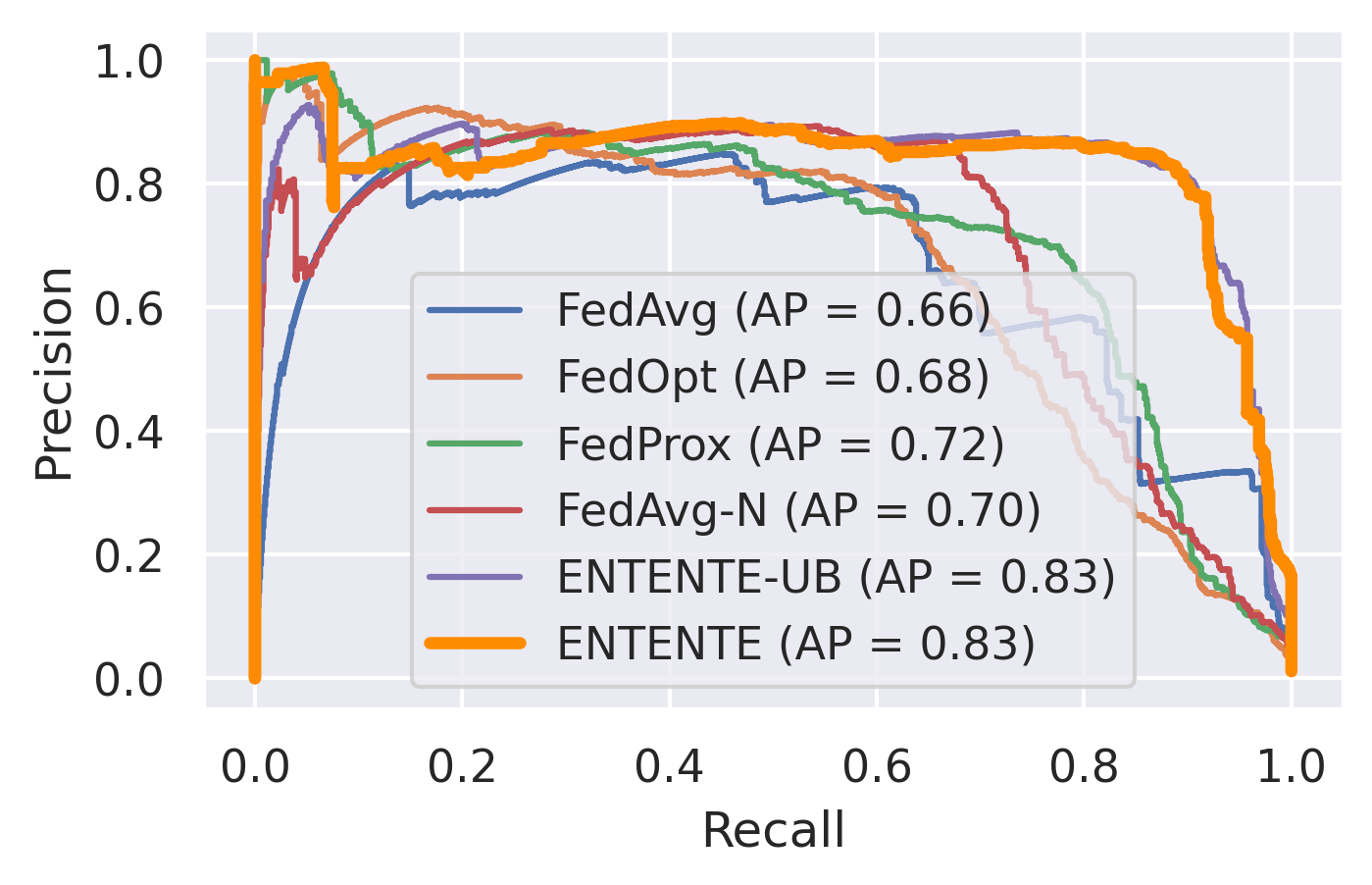}
    \caption{The precision-recall curve under different FL methods on OpTC dataset ($K=3$).
    }
    \label{fig:euler_optc}
    \vspace{-0.1in}
\end{figure}

In Figure~\ref{fig:euler_optc}, we draw the precision-recall curve of different FL methods with one client number ($K=3$) due to space limit. %
The precision of \system outperforms the other systems at most recall values, and it can reach 0.8 precision at about 0.9 recall.

\zztitle{Results on LANL.}
In Table~\ref{tab:whole}, we show AP and AUC when LANL is tested with Euler. Since Euler uses the redteam events as TP, the class distribution is highly imbalanced, and all systems have low AP\footnote{The paper of Euler~\cite{king2023euler} claims 5.23\% (GCN+GRU models) AP on LANL, but we cannot reproduce the same result with their GitHub code~\cite {euler-github}. A few possible explanations include we use different random seeds and different hardware (we use GPUs while Euler uses CPUs).}. 
\system and \system-UB do not outperform non-FL but are still the ones closest to non-FL in all cases ($K=2-5$) in terms of AP.
We notice that the trend of AP is somehow opposite to AUC (e.g., AP of \system is higher than FedAvg when $K=3,5$) and this observation can be attributed to the differences in the training and testing data distribution. 
In particular, Euler uses negative sampling to ``synthesize'' malicious edges during training. However, the ground-truth malicious edges during testing have a much smaller volume (only hundreds), which might not be characterized well by the sampled negative edges. Correctly classifying malicious edges is a major goal modeled under AP, but not necessarily so under AUC, when the vast majority of edges are benign.
We acknowledge that the AP is still low, and discuss this limitation in Section~\ref{sec:discussion}. Yet, we argue that the \textit{improvement over the existing systems} is more important, and \system achieves this goal.

Table~\ref{tab:whole} also shows the result when Jbeil is the GNIDS. This time, as non-existent edges are considered as TP (explained in Section~\ref{subsec:settings}), higher AP is observed. We found \system still outperforms the other baselines (except $K=3$, lower than \systemNumber). Though \system performs worse than \system-UB, the margin is small.
The more balanced class distribution makes non-FL the winner in most cases.

\zztitle{Results on Pivoting.} 
Due to space limit, we present the results in Appendix~\ref{app:lanloptc}. In short, we found \system achieves the best AP and AUC in the majority of the cases among the FL methods.

\subsection{Scalability}
\label{subsec:scalability}

Following the scalability requirement described in Section~\ref{subsec:goals}, we measure the communication overhead, latency and memory consumption caused by \system. 

We first follow the default data split setup ($K=2-5$). Regarding latency, when training LANL and OpTC with Euler, the whole training process takes an average of 116.33 
seconds and 706.40 seconds (30 FL iterations till the training converges for LANL and OpTC) and the testing process takes 160.99 seconds and 0.6139 seconds\footnote{OpTC has much larger training latency and much smaller testing latency than LANL because only 23\% of whole OpTC data is used for testing (while 96\% for LANL).}.
For Pivoting, the training process takes an average of 314.58 seconds (5 FL iterations) and the testing process takes 45.963 seconds. 

Compared to basic FedAVG, \system introduces additional overhead from generating the global reference graph and updating client weights using WLH and ACS. On LANL with $K=4$, the first step takes 5.47, 7.90, 2.70, and 0.24 seconds across the four clients, while the second step takes only 0.082 seconds. On OpTC, the first step takes 0.54, 2.25, 0.89, and 25.68 seconds, and the second step takes 0.0058 seconds. On Pivoting, the first step takes 0.049, 0.009, 0.082, and 0.035 seconds, and the second step takes 0.051 seconds.
The differences between clients are due to their different sizes. 
Overall, while \system adds some computation time, the overhead remains minor relative to the full training process.

By using FL, \system also reduces the data transmitted from client to server when training a GNIDS model, saving clients' bandwidth. Take Euler+LANL as an example. The central server would need to access the 65GB network logs to train GNIDS. With \system, only model weights $w^k_i$ are transmitted by a client $k$ in an FL iteration $i$, and only \textit{1.94MB} model weights are transmitted in total. 
\revisednew{Then, we estimate the per-iteration latency for LANL+Euler ($K=4$) considering network transmission overhead. Assuming the client upload link bandwidth is 1 MB/s following the assumption of the FedAvg paper~\cite{mcmahan2017communication} and the download link bandwidth is 10 MB/s. The server has much larger bandwidth and it will not throttle clients' network transmission. The max client computation time per iteration is 7 s. The latency per iteration will be $\frac{1.94}{1} + \frac{1.94}{10} + 7 = 9.134 s$. The latency is expected to be similar for larger $K$ if server bandwidth is much larger than client's.}

\begin{table}[h]
    \caption{\revisednew{Performance of LANL+Euler under large $K$ (from 5 to 20). ``Valid Client\#'' means the number of clients containing nodes after MBM clustering. ``Training'' is the latency of the whole training process. ``CPU'' and ``GPU'' are the averaged CPU and GPU memory overhead per client. 
    }}
    \centering
    \begin{tabular}{c|cccccc}
    \hline
        Valid Client\# & Training(s) & AP(\%) & AUC(\%) & CPU(MB) & GPU(MB) \\
        \hline
        5 & 138.51 & 0.67 & 97.16 & 1521.21 & 3343.39 \\
        9 & 224.38 & 0.77 & 98.72 & 866.13  & 3256.13\\
        20 \tablefootnote{With $K=20$, the clustering process results in only 15 valid, non-empty clients.} & 497.46 & 1.05 & 99.12 & 522.87 & 3308.80\\
        50 \tablefootnote{With $K=50$, the clustering process results in only 35 valid, non-empty clients. This experiment was executed on a CPU due to GPU memory limitations.} & 5111.43 & 0.20 & 97.99 & 1341.30 & N/A \\
        \hline
    \end{tabular}
    \label{tab:scale}
\end{table}

Finally, we attempt to increase $K$ to be more than 5 and test whether \system can scale up to many clients, and show the results in Table~\ref{tab:scale}. We evaluate this setting on the LANL dataset because it contains a large number of nodes (17,649). For OpTC and Pivoting datasets, increasing $K$ will yield invalid clusters from MBM as discussed in Section~\ref{subsec:settings}. Yet, for large $K$ (e.g., 10, 20, and 50), we found MBM generates empty clusters, which impedes generating more valid clients, so we treat the non-empty clusters as clients (9, 15, 35 for $K=\{10, 20, 50\}$).
We found that the training time does increase notably when $K=9$ and $K=15$, but this is because the data loading time is longer (more cross-client edges appear when $K$ increases). Besides, since we simulate all clients on a single machine and they share a GPU, the resource contention also elongates the training time. The training time is expected to reduce when the clients operate in a real-world distributed environment.
With $K=50$, due to GPU memory limitation, we only test on the CPU environment. 
AP and AUC are also recorded, and we find that AP and AUC get better when $K$ is set to medium values and drop again when $K$ is increased to 50.
\revisednew{
We further measure memory consumption introduced by \system for a larger number of clients to assess its per-client overhead. For each run of \system, we measure CPU and GPU overhead separately and compute the average consumption across clients. 
Specifically, the CPU memory overhead per client drops from 1521 MB to 522 MB when $K$ rises from 5 to 20, which is caused by the reduction of the graph size processed by each client. 
For $K=50$, the average CPU usage is smaller compared to $K=5$ (1521.21 MB vs 1341.30 MB). We speculate the overhead is mainly caused by the CPU–GPU transfer buffers and data loading since only the CPU is used in the setting.
GPU memory usage remains relatively stable (around 3.2–3.3 GB), suggesting that the GPU memory footprint is largely determined by the model architecture. 
}

\ignore{
\subsection{Efficiency}
\label{subsec:efficiency}

\zl{Like Argus, measure the overhead (GPU is fine). I think we should highlight our graph similarity computation based on WL kernel is quite efficient.}\jxu{How to find a slow way to calculate the graph similarity to show the improvement of WL kernel?}
\begin{table}[]
    \centering
    \caption{Training and testing time of \system{} under different client numbers for LANL.}
    \begin{tabular}{ccc}
    \hline
    Weight & Training (s) & Testing (s)\\
    \hline
        
    Node Number & \\
    \system &  \\

    \hline        
    \end{tabular}
    \label{tab:time_cost}
\end{table}
}

\subsection{Robustness against Poisoning Attacks}
\label{subsec:poison}

In this subsection, we evaluate how \system performs when FL poisoning attack is conducted by a compromised client. Recently, a few adaptive attacks against GNIDS were developed~\cite{xu2023cover, goyal2023sometimes, mukherjee2023evading}. Only Xu et al.~\cite{xu2023cover} tested the training-time attacks (the other focused on testing-time attacks).
So far, none of the prior works considered the robustness of a GNIDS model trained under FL, and we adapt the attack from Xu et al.~\cite{xu2023cover} to our setting, which has been evaluated against Euler on the LANL dataset.

In the original attack of Xu et al.~\cite{xu2023cover}, \textit{covering} accesses are added for each malicious edge to avoid exposing the malicious edges when GNIDS is trained.
In our setting, the attacker can directly inject the malicious edges to the logs \textit{after} they are collected by the FL client, under the model poisoning adversary~\cite{bagdasaryan2020backdoor}.
As such, the covering edges are unnecessary. 
In addition, the attacker can scale up the model updates by a constant factor $\gamma$ to outweigh the updates from other benign clients~\cite{bagdasaryan2020backdoor}.

We present the attack pseudo-code in Algorithm~\ref{alg:poison}. %
Specifically, on a controlled client $k$, the attacker needs to enumerate all the sub-graphs $[\mathcal{G}_1^k, \dots, \mathcal{G}_T^k]$ and compare the source and destination nodes to its malicious edges $\mathcal{EM}$ to be used in the testing time. An edge will be added when there is a pair-wise match. We also use a likelihood threshold $p$ to control the number of injected edges.

\begin{algorithm}[h]
\caption{The proposed poisoning attack. ClientUpdate is the same as Algorithm~\ref{alg:system}. Each client's graph $\mathcal{G}^k = (\mathcal{V}^k, \mathcal{E}^k)$
}
\label{alg:poison}

\KwData{ Number of subgraphs \( T \); Malicious edges \( 
 \mathcal{EM} \); Likelihood threshold $p$}

\BlankLine

\ForEach{client  $k$ from the malicious clients}{

\For{$t\gets1$ \KwTo T }{
            \If{random number $<$ p}
            {
                
               \For{$e \in \mathcal{EM}^k$}{
                \If{$e \not\in \mathcal{E}^k_t$ \textbf{and} $e.src \in \mathcal{V}^k_t$ \textbf{and} $e.dst \in \mathcal{V}^k_t$}{
                    $\mathcal{G}^k_t.add(e)$\;
                    } 
                }
            }
        }
Send $\gamma$ $\times$ ClientUpdate($k, w_i$) to central server in each FL iteration
}
\end{algorithm}

Since the poisoning attack conducted by Xu et al~\cite{xu2023cover} is on Euler+LANL, we focus on attacking the Euler GNIDS. We test both LANL and OpTC datasets to evaluate the generality of the attack. 
In Table~\ref{tab:attack_lanl}, we show the attack result when Euler is trained on LANL under \system. 
We select client 4 as the malicious client, which observes 495 malicious edges (out of 517 total malicious edges). 
Among these edges, 492 are cross-client edges, so the attacker has the motivation to poison the global model to hide their attack from the other clients. 
We define ``success rate'' (SR) as the ratio of malicious edges that evade detection, which is similar to Xu et al.~\cite{xu2023cover}. We also count the number of injected edges per malicious edge (EPM). EPM could be larger than 1 when multiple snapshots are poisoned. 
\ignore{
Like~\cite{xu2023cover}, we fix the classification threshold $\tau$ of the Euler model under different $p$\footnote{We obtained the source code of ~\cite{xu2023cover} from the authors and we set $\tau=0.5$ based on their source code.}. 
\zl{remove this sentence, i think it's wrong to generate AUC and AP under different $\tau$. $\tau$ is already fixed.}
}
$\gamma$ is set to values between 5 and 100 and $p$ is set to values between 25\% to 100\%. The row with ``-'' in $\gamma$ and $p$ in Table~\ref{tab:attack_lanl} shows result without attack. 
The attack is more powerful with larger $\gamma$ and $p$. The result shows \system can bound SR to a low number (9.30\%)\footnote{The implementation of Xu et al. sets the classification threshold of Euler manually and then injects edges according to the threshold. This setting differs from Euler's default setting which learns the classification threshold from a validation set. We follow Euler's default setting, which might lead to worse performance than Xu et al.}. When disabling the norm bounding defense (\system-UB), the training GNIDS model observes higher SR, but more importantly, ``NaN'' (not a number) error occurs with larger $p$ and $\gamma$. In this case,  the gradient updates accumulate substantial changes, leading to sudden large adjustments of the global model weights and gradient explosion.
Training the GNIDS model will fail, which maps to the untargeted attack against FL~\cite{fang2020local}.

\begin{table}[h]
    \centering
    \caption{Attack on LANL ($K=4$). Client 4 is malicious. The upper (``E'') and lower parts (``EUB'') of the table shows the result of \system and \system-UB.
    }
    \begin{tabular}{l|cccccc}
    \hline
        & $p$ & $\gamma$ & AP(\%) & AUC(\%) & SR(\%) & EPM \\     
        \hline
        \multirow{6}{*}{E} & - & - & 0.72 & 97.00 &  - \\
        & 25\% & 5 & 0.54 & 96.96 & 9.30 & 22\\
        & 50\% & 5 & 0.56 & 96.97 &  9.30 & 40\\
        & 50\% & 25 & 0.56 & 96.97 & 9.30 & 40 \\
        & 75\% & 25 & 0.60 & 96.99 & 9.30 & 62 \\
        & 100\% & 100 & 0.55 & 96.96 & 9.30 & 80 \\
        \hline       
        \multirow{3}{*}{EUB} & 25\% & 5 & 0.74 & 97.34 & 11.4 & 22 \\  
        & 50\% & 25 & 0.18 & 96.23 & 40.59 & 22 \\   
 
        & 100\% & 100 & NaN & NaN & - & 80 \\   
        \hline
    \end{tabular}
    \label{tab:attack_lanl}
\end{table}

\revisednew{
Next, we change the malicious client from client 4 to client 1, 2, and 3 separately, and evaluate the impact. We also consider the setting of colluding clients and evaluate \system when both client 1 and 4 are malicious. We fix the attack parameters that $\gamma$=100, $p$=100\%, and ``EPM''  to 80 to maximize attacker's capabilities.
The result on LANL is shown in Table~\ref{tab:malicious_client_number_lanl}. 
SR is still low and we found client collusion does not lead to higher SR against \system. 
}

\begin{table}[h]
    \small
    \centering
    \caption{\revisednew{Different malicious client on LANL.}
    }
    \begin{tabular}{c|ccccccc}
    \hline
        & Client & AP (\%) & AUC (\%) & SR(\%) \\
        \hline
        \multirow{4}{*}{E} &1 & 0.37 & 96.54 & 6.60 \\
        & 2 & 0.53 & 96.93 & 10.48 \\
        & 3 & 0.55 & 96.96 & 8.9 \\
        & 1, 4 & 0.58 & 96.90 & 7.48 \\
        \hline
    \end{tabular}
    \label{tab:malicious_client_number_lanl}
\end{table}

Regarding OpTC, we choose $K=3$ and use client 3 as a malicious client. Similar as the result on LANL, \system is robust even under very large $p$ and $\gamma$, as shown in Table~\ref{tab:attack_optc}. Without the norm-bounding defense, much higher SR is observed or the training process cannot finish.

\begin{table}[h]
    \centering
    \caption{Attack on OpTC ($K=3$). Client 3 is malicious. The upper (``E'') and lower (``EUB'') parts of the table shows the result of \system and \system-UB.
    }
    \begin{tabular}{l|ccccccccccc}
    \hline
        & $p$ & $\gamma$ & AP(\%) & AUC(\%) & SR(\%) & EPM \\
        \hline
        \multirow{5}{*}{E} & - & - & 83.29 & 99.76 &  - \\        
        & 10\% & 2 & 83.70 & 99.72 & 17.65 & 148\\
        & 25\% & 5 & 83.48 & 99.76 & 14.61 & 374\\
        & 50\% & 10 & 85.69 & 99.78 &  13.38 & 764\\
        & 100\% & 100 & 84.91 & 99.75 &  14.20 & 1513\\
        \hline
        \multirow{2}{*}{EUB} & 10\% & 2 &  72.94 &  99.30 & 28.24 & 148\\
        & 50\% & 5 & NaN & NaN & - & 764\\
        & 100\% & 25 & NaN & NaN &  - & 1513\\
        
        \hline
    \end{tabular}
    \label{tab:attack_optc}
\end{table}

\subsection{Ablation Study}
\label{subsec:ablation}

\revisednew{Due to space limit, we provide the results in Appendix~\ref{app:ablation}, including the impact of ACS, graph augmentation, clients' weight initialization, $K$ and $m$, comparison between local and global models, and visualization of MBM clustering.}

\section{Discussion}
\label{sec:discussion}

\zztitle{Limitations and future works.} 
First, we admit that the results of \system on LANL are far from ideal when the redteam events are TP. 
A few reasons have been given in prior works~\cite{king2023euler,xu2024understanding}, including: 1) the labeled malicious events are too coarse-grained; 2) the malicious activities after the redteam's ground-truth events are not tracked; 3) potentially malicious events on included in the ``attack-free'' training period. 
Second, due to the lack of ground truth of the locations of each device in the LANL or OpTC dataset, we choose to run the clustering method MBM~\cite{larroche2022multilayer} to generate FL clients' data. 
The location information is usually not provided from a log dataset and the IP is also anonymized. Clustering is our best effort. Similar approaches have been followed by other FGL works like ~\cite{yao2022fedgcn, zhang2021subgraph, xie2021federated}.
Third, we did not simulate \system and other baseline FL methods in a fully distributed environment, i.e., different clients on different machines. Hence, the actual overhead, especially network communication, should be higher. However, the extra overhead caused by \system over the basic FedAVG method is introduced during bootstrap and ACS, and Section~\ref{subsec:effectiveness} shows it is reasonable. 
Finally, we did not simulate different FL poisoning attacks (surveyed in Section~\ref{sec:related}) besides Xu et al.~\cite{xu2023cover}, due to the gap between the attack methods and the concrete GNIDS setting. We leave the exploration of new attacks as future work.

\zztitle{Privacy implications of \system.}
\revisednew{Due to the usage of FL, the privacy attacks against FL, like membership inference attack (MIA)~\cite{nasr2019comprehensive} and gradient inversion attack (GIA)~\cite{wang2019beyond}, can be utilized to attack \system.} 
To safeguard the privacy of FL clients, differential privacy (DP) can be applied to add noise to the model updates, which hides the existence of a single instance under FL. \revisednew{As such, DP directly defends against MIA~\cite{xie2023unraveling}. 
Hatamizadeh et al. evaluated DP against GIA, and the results show DP-SGD is particularly effective in reducing the leakage from gradients~\cite{hatamizadeh2023gradient}.}
Besides, poisoning attacks can be deterred by DP, as confirmed by previous studies 
~\cite{sun2019can, xie2023unraveling,geyer2017differentially,icml/ZhangCH0Y22}. 
We have conducted a preliminary analysis regarding the combination of DP and \system. As detailed in Appendix~\ref{app:dpfl}, 
our experimental results suggest that achieving strong defense against poisoning while preserving acceptable model utility remains a significant challenge. 

\revisednew{
\zztitle{Other options for cross-silo GNIDS.}
Cryptographic tools such as multiparty computation (MPC) and homomorphic encryption (HE) can be used to train a model using data from different regions with strong privacy protection. However, MPC and HE would incur much higher communication and computational overheads than FL, making them unsuitable to train GNIDS on large-scale log data now. CoGNN~\cite{ccs/ZouLSLXX24}, the SOTA approach for MPC-based graph learning, communicates from 1GB to 4GB \textit{per epoch} on small-scale datasets like Cora (2,708 nodes and 10,556 edges). The overhead will be amplified under larger GNIDS datasets. Meanwhile, \system transmits 1.94MB model weights in total when training Euler on LANL, as described in Section~\ref{subsec:scalability}. 
Moreover, CoGNN only supports static GNN models, but the SOTA GNIDS explicitly leverage temporal information with non-GNN models like GRU (integrated by Euler) and TGN (integrated by Jbeil). 
As a result, we believe FL is the most suitable approach for our problem setting at the current stage. 
}

\section{Related Work}
\label{sec:related}

\zztitle{Host-based Intrusion Detection (HIDS).} 
This work focuses on developing privacy-preserving GNIDS, which consume network logs. Graph learning has also been applied to host logs. We refer interested readers to a survey by Inam et al. ~\cite{inam2023sok} for details. Below we provide a brief overview. The main approach to detect intrusions under HIDS is through \textit{provenance tracking}~\cite{king2005backtracking, jiang2006provenance}, which performs variations of breadth-first search (e.g., under temporal constraints like happens-before relation~\cite{shu2018threat}) to find other attack-related nodes given Indicators of Compromise (IoCs). 
Yet, provenance tracking often leads to a large number of candidate nodes to be investigated~\cite{xu2016high}, and many HIDS add heuristics 
to reduce the investigation scope~\cite{hassan2019nodoze, liu2018towards, milajerdi2019holmes, hossain2020combating}. Another way to reduce the false positives is to simplify the graphs, e.g., through event abstraction~\cite{hassan2020tactical, yu2021alchemist} and graph summation~\cite{xu2022depcomm}. 

Recently, like GNIDS, GNN has also been tested for HIDS and various embedding techniques have been developed/used, including masked representation learning~\cite{jia2023magic}, TGN~\cite{cheng2024kairos}, Graph2Vec ~\cite{yang2023prographer}, embedding recycling~\cite{rehman2024flash}, root cause embedding~\cite{goyal2024r}, etc. It is likely that FL could benefit these HIDS, but a different operational model and/or learning method is needed. For instance, cross-device FL instead of cross-silo FL is a more suitable FL setup.

\zztitle{FL for security.} 
Before our work, FL has been applied in various security-related settings, like risk modeling on mobile devices~\cite{fereidooni2022fedcri}, malicious URL detection~\cite{khramtsova2020federated, ongun2022celest}, detecting abnormal IoT devices~\cite{nguyen2019diot}, malware detection~\cite{rey2022federated}, browser fingerprinting detection~\cite{annamalai2024fp}, etc.
\revisednew{Though FL has been used to train a model to detect intrusions with system logs collected from cloud instances~\cite{de2022interpretable} and security events collected from participating organizations~\cite{naseri2022cerberus}, the models trained by these works (i.e., Transformer and RNN) are tailored to the tabular representation from logs. So far, none of the existing works have trained a GNN model with FL to detect intrusions from large-scale network logs.}

\zztitle{Security and privacy for FL.} 
Here, we present a more detailed overview of security and privacy issues in FL. Regarding security, poisoning (or backdoor) attacks are considered as the major threat~\cite{mothukuri2021survey}, and the previous attacks can be divided into data poisoning~\cite{nguyen2020poisoning, shen2016auror, xie2019dba} that manipulates the clients' training data and model poisoning~\cite{bagdasaryan2020backdoor, wang2020attack, li20233dfed, krauss2024automatic} that manipulates the training process or models themselves. In either case, poisoning will result in deviation of model updates, and the majority of defenses aim to detect and mitigate abnormal deviation~\cite{rieger2022deepsight, nguyen2022flame,kumari2023baybfed, fereidooni2024freqfed, sun2019can,rieger2024crowdguard, kabir2024flshield}. \system integrates norm bounding~\cite{sun2019can} into the training procedure and the empirical and theoretical results show its effectiveness.

To defend FL against inference attacks~\cite{nasr2019comprehensive,wang2019beyond}, differential privacy (DP) has been applied to FL. The noises can be added to each training step with DP-SGD~\cite{sun2021ldp, truex2020ldp}, to the trained local model~\cite{agarwal2021skellam, kairouz2021distributed} and to the central server~\cite{geyer2017differentially, mcmahan2018learning}. Recently, Yang et al. studied the accuracy degradation caused by FL+DP, and exploited client heterogeneity to improve the FL model's accuracy~\cite{yang2023privatefl}. We attempted to integrate DP into \system but the accuracy drop is prominent. We also believe the privacy threat and the privacy notions need to be clearly defined under the GNIDS setting, in order to enable more effective defense.

\section{Conclusion}

In this paper, we propose \system, a new Graph-based Network Intrusion Detection Systems (GNIDS) backed by Federated Learning (FL), to address concerns in data sharing. We carefully tailor the design of \system to the unique settings of the security datasets (e.g., highly imbalanced classes and non-IID client data) and variations of GNIDS tasks/architectures. With new techniques like weight initialization and adaptive contribution scaling, we are able to achieve the desired design goals (effectiveness, scalability and robustness) altogether. The evaluation of three large-scale log datasets, namely OpTC, LANL and Pivoting, shows that \system outperforms the other baseline systems, even including the model trained on the whole dataset in some cases. 
We also conduct adaptive attacks against \system with FL poisoning, and show that \system can bound the attack success rate and ensure the training procedure can finish as desired.
With \system, we hope to encourage more research to address the data sharing issues in building GNIDS, which is understudied, and develop better attack and defense benchmarks to evaluate the robustness of federated graph learning (FGL) systems.

\section*{Acknowledgment}
We thank reviewers for the valuable suggestions. 
This work is supported by NSF CNS-2220434 and Amazon Research Awards. Any opinions, findings, and conclusions or recommendations expressed in this material are those of the author(s) and do not reflect the views of the sponsors.

\bibliographystyle{IEEEtran}
\bibliography{main}

\appendix
\normalsize

\subsection{Algorithm of BA Model}
\label{app:ba}

In Algorithm~\ref{alg:ba}, we describe BA model with pseudo-code.

   \begin{algorithm}[h!]
\caption{Barabási–Albert (BA) Model}
\label{alg:ba}
\KwData{$n$ (number of nodes), $m$ (number of edges to attach from a new node to existing nodes)}
\KwResult{$G$ (generated graph)}
\BlankLine
$G \gets \text{InitializeGraph}(m)$\;
\For{$i = m+1$ \KwTo $n$}{
    $\text{AddNode}(G, i)$\;
    $k_j \gets \text{Degree}(G)$\;
    $K \gets \sum_{j \in G} k_j$\;
    $L \gets \text{EmptyList}()$\;
    \ForEach{$j \in G$}{
        $L.\text{Add}(j, k_j / K)$\;
    }
    $S \gets \text{RandomSelect}(L, m)$\;
    \ForEach{$v \in S$}{
        $\text{AddEdge}(G, i, v)$\;
    }
}
\Return $G$
\end{algorithm}

\subsection{Workflow of \system}
\label{app:workflow}
In Algorithm~\ref{alg:system}, we summarize the workflow of \system.
Specifically, $w_i$ contains the model parameters for the graph encoder, temporal encoder and decoder, which are denoted by  $\text{ENC}(\cdot)$, $\text{TEMP}(\cdot)$ and $\text{DEC}(\cdot)$. For the decoder, only the trainable implementations go through FL. $\text{JS}(\cdot)$ is Jaccard similarity. $\nabla \mathcal{L}$ computes gradients on the training loss, and the same loss function of the GNIDS is used.

\ignore{
\begin{algorithm} 
\caption{\system{} training procedure. \zl{remove RNN, make it more generic}}
\KwData{Graph $\mathcal{G}$, Node partitions $\mathcal{P}$}
\KwResult{Global model $M$(GNN, RNN, DEC)}
$\mathcal{G}_0, \mathcal{G}_1, \dots, \mathcal{G}_T =\text{Separate}(\mathcal{G}, \mathcal{T})$\;
\For{each partition $P_i \in \mathcal{P}$}{
    Initialize local model $M_i$ on $P_i$\;
}
\While{not converged}{
    \For{each partition $\mathcal{G}(P_j) \in \mathcal{G}(\mathcal{P})$}{
            $loss=0$\;
        \For{$i\gets0$ \KwTo T }{
            $X_i^{train}, A_i^{train} = \text{data\_split}(\mathcal{G}_i(\mathcal{P}_j))$\;
            $Z'_i = \text{GNN}_j(X_i^{train}, A_i^{train})$\;
            $Z_i = \text{RNN}_j(Z'_i)$\;
            $\hat{A}_i = \text{DEC}_j(Z_i)$\;
            $\text{loss} = \text{loss} + L_{ce}$\;  
        }
        loss.$\text{back\_propagate}$\;
    }
    Aggregate local models to update the global model $M$(GNN, RNN, DEC)\;
}
\end{algorithm}
}

\begin{algorithm}[h!]
\caption{Global model training under \system. 
}
\label{alg:system}

\KwData{\( E \) is the number of local epochs; \( \eta \) is the learning rate;  Number of maximum FL iterations \( R \); Number of subgraphs \( T \); $c_1$ and $c_2$ are two constants.}
\KwResult{Global model parameters \( w_{i+1}  \)}

\BlankLine

\textbf{Server executes:}\\
Initialize \( w_1 \)\; 
\(\mathcal{G}^{ref}\)  \( \gets \) BA\_Model(\(n\), \(m\))\; 
\ForEach{client \( k \) from all clients}{
        \(S_{Jac}^k\)  \( \gets \)  ClientInitialWeight(\( k, \mathcal{G}^{ref} \))\;
}
\For{\( i = 1 \) \KwTo \( R \)}{
    \ForEach{client \( k \) from all clients}{
        \( w^k_{i} \) \( \gets \) ClientUpdate(\( k, w_i \))\;
    }
    \(S^K_i, D^K_i = \) ACS(\( w_i, w^k_i \))\;
    \( w_{i+1}[\text{ENC}] \) \( \gets \) \( w_i[\text{ENC}] \) +  \( \Sigma^{K}_{k=1} (c_1 \times S^k_{Jac} + c_2 \times  S^k_i \times D^k_i) \times \mathsf{NB}(\Delta w^k_{i+1}[\text{ENC}]) \) \;
    \( w_{i+1}[\text{TEMP}] \) \( \gets \)  \( w_i[\text{TEMP}] + \Sigma^{K}_{k=1} (c_1 \times S^k_{Jac} + c_2 \times  S^k_i \times D^k_i) \times \mathsf{NB}(\Delta w^k_{i+1}[\text{TEMP}] \) \;
    \If{ \text{DEC} is trainable}{
        \( w_{i+1}[\text{DEC}] \) \( \gets \)  \( w_i[\text{DEC}] + \Sigma^{K}_{k=1} (c_1 \times S^k_{Jac} + c_2 \times  S^k_i \times D^k_i) \times \mathsf{NB}(\Delta w^k_{i+1}[\text{DEC}] \) \;
    }
    
    \If{\( \text{early\_stopping(\( w_1, \dots, w_{i+1} \))} \)}{
        \textbf{break}\;
    }
}
\Return \( w_{i+1}  \)

\BlankLine

\SetKwFunction{FClientWeight}{ClientInitialWeight}

\SetKwProg{Fn}{Function}{:}{}
\Fn{\FClientWeight{\( k, \mathcal{G}^{ref} \)}}{
    Generate local graph \(\mathcal{G}^k\)\;
    Add cross-client edges to \(\mathcal{G}^k\)\; 
    \( [\mathcal{G}_1^k, \dots, \mathcal{G}_T^k] \) \( \gets \)  separate(\(\mathcal{G}^k\))\;    
    \Return JS(WLH(\(\mathcal{G}^k\)), WLH(\(\mathcal{G}^{ref}\))) to server\;
}

\BlankLine
\SetKwFunction{FACS}{ACS}

\SetKwProg{Fn}{Function}{:}{}
\Fn{\FACS{\(w_i, w^k_i \)}}{
    Compute \(S^K_i\) and \(D^K_i\) \;
    \Return \(S^K_i, D^K_i\) to server\;
}

\SetKwFunction{FClientUpdate}{ClientUpdate}

\SetKwProg{Fn}{Function}{:}{}
\Fn{\FClientUpdate{\( k, w \)}}{
    \For{each local epoch $i$ from 1 to E}{
         $w \gets w - \eta \nabla \mathcal{L}(w;[\mathcal{G}_1^k, \dots, \mathcal{G}_T^k] ) $\;
    }
    \Return \( w \) to server\;
}
\end{algorithm}

\subsection{Datasets and pre-processing}
\label{app:datasets}

The OpTC dataset~\cite{optc} contains the telemetry data collected under the DARPA TC program~\cite{darpa-tc}, during which APT attacks were simulated on different OSes. The host-level activities between subjects like processes and objects like files and sockets were logged. Like Euler~\cite{king2023euler}, we use the ``START'' events related to the ``FLOW'' objects (i.e., network flows). The nodes are hosts distinguished by IP addresses and the flows between hosts are merged into edges. The statistics after this step are shown in Table~\ref{tab:dataset_all}. We split the data into 6 minutes (360s) window.
We use the first 5 days' snapshots (no redteam events exist) for training and the remaining 3 days' snapshots for testing. OpTC is only tested by Euler and we report its result accordingly.

The LANL dataset~\cite{lanl-ds-15} contains anonymized event data from four sources within Los Alamos National Laboratory's internal computer network. We use the authentication logs from individual computers and domain controller servers following~\cite{king2023euler,khoury2024jbeil}.
Simulated redteam attacks are conducted from the compromised machines. As shown in Table~\ref{tab:dataset_all}, the malicious events consist of a fairly small portion among all events, posing a great challenge to GNIDS. 
Regarding data pre-processing, like ~\cite{king2023euler,khoury2024jbeil}, we only keep the events with the keyword ``NTLM''~\cite{ntlm}, as other events are unrelated to authentications. For Euler, the logs are split by a 30-minute (1,800s) window into snapshots.
We use ``source computer'' and ``destination computer'' as nodes and all events sharing the same pairs of nodes are merged into an edge. The first 42 hours are used to train the Euler GNIDS. On average, each snapshot has 7,957 edges and we use 5\% edges for validation. After 42 hours, redteam events appear and the following snapshots are used for testing. 
For Jbeil, we follow their pre-processing procedures to use the events that happened during the first 14 days, including 12,049,423 events. 
Different from Euler, Jbeil did not use the LANL's redteam events as malicious samples. Instead, it injects non-existent edges (i.e., negative sampling) and conducts link prediction (i.e., predicting an edge's existence in testing)\footnote{Jbeil has another mode that uses a lateral-movement simulator to inject malicious edges, but this simulator has not been released. We have contacted the authors and confirmed it.}. We use 70\% events for training, 15\% for validation and the remaining 15\% for testing. For the transductive learning mode, the training and testing sets share the same set of nodes. For the inductive learning mode, 30\% of the nodes are hidden during the training but unmasked during testing, following Jbeil's setting.

The Pivoting dataset~\cite{apruzzese2017detection} comprises network flow data, collected over a full working day under a large organizational setting, exclusively representing internal-to-internal communications among hosts. The simulated pivoting activities for lateral movement are considered malicious. 
It is only tested by Jbeil
and we follow its setting and use the first 1,000,000 events, which correspond to 698 nodes. Again, the data is split into 70\% for training, 15\% for validation, and the remaining 15\% for testing. Link prediction is conducted and we test both transductive and inductive learning modes.

\subsection{More Results on LANL, OpTC and Pivoting Dataset}
\label{app:lanloptc}

Table \ref{tab:comparison_jbeil_pivoting} shows the results on the Pivoting dataset when Jbeil is the GNIDS. 
We observe that \system achieves the best AUC in most cases ($K=2,4$ for transductive mode and $K=3,4$ for inductive mode). FedOpt outperforms \system when $K=2$ and \system-UB outperforms when $K=3$ for transductive but the differences are small (less than 5\%). Though non-FL achieves the best AP and AUC in nearly all cases, the margins over \system are also small (less than 7\% AP and 4\% AUC).%

\begin{table*}[h]
    \small
    \centering
    \caption{Evaluation on Pivoting with Jbeil. Non-FL is listed for ease of comparison. All numbers are used in the percent format. The \colorbox{colorbest}{\color{colortext}{best}} and the \colorbox{colorsecond}{\color{colortext}second-best} among FL methods are highlighted respectively. \greenup means outperforming non-FL and \reddown means the worst performance.}
    \begin{tabular}{l|cc|cc|cc|cc|cc|cc}
    \hline
    \multicolumn{7}{c|}{\textbf{{\small Pivoting-Jbeil-Transductive}}} & \multicolumn{6}{c}{\textbf{{\small Pivoting-Jbeil-Inductive}}} \\
    {\small Client\#} & \multicolumn{2}{c}{\textbf{2}} & \multicolumn{2}{c}{\textbf{3}} & \multicolumn{2}{c|}{\textbf{4}} & \multicolumn{2}{c}{\textbf{2}} & \multicolumn{2}{c}{\textbf{3}}& \multicolumn{2}{c}{\textbf{4}} \\
    \hline
    {\small Algorithm} & AP & AUC & AP & AUC & AP & AUC  & AP & AUC & AP & AUC & AP & AUC \\

    Non-FL       & 96.26$\tbspace$ & 97.05$\tbspace$ & 96.26$\tbspace$ & 97.05$\tbspace$ & 96.26$\tbspace$ & 97.05$\tbspace$ & 95.92$\tbspace$ & 96.87$\tbspace$ & 95.92$\tbspace$ & 96.87$\tbspace$ & 95.92$\tbspace$ & 96.87$\tbspace$ \\
    \hline
    FedAVG       & 82.07$\tbspace$ & 88.70$\tbspace$ & 68.81\reddown & 75.70\reddown & 81.19$\tbspace$ & 87.26$\tbspace$ & 83.59$\tbspace$ & 89.46$\tbspace$ & 71.05\reddown & 78.06\reddown & 80.80$\tbspace$ & 86.99$\tbspace$ \\
    FedOpt       & \cellcolor{colorbest}\color{colortext}{95.81}$\tbspace$ & \cellcolor{colorsecond}\color{colortext}97.02$\tbspace$ & 81.14$\tbspace$ & 87.37$\tbspace$ & 90.28$\tbspace$ & 93.97$\tbspace$ & \cellcolor{colorbest}\color{colortext}{96.21}\greenup & \cellcolor{colorbest}\color{colortext}{97.44}\greenup & 85.70$\tbspace$ & 90.87$\tbspace$ & 91.65$\tbspace$ & 94.94$\tbspace$ \\
    FedProx      & 81.01$\tbspace$ & 87.34$\tbspace$ & 80.72$\tbspace$ & 86.71$\tbspace$ & 58.62\reddown & 62.56\reddown & 80.58$\tbspace$ & 87.17$\tbspace$ & 81.53$\tbspace$ & 87.40$\tbspace$ & 58.08\reddown & 61.99\reddown \\
    \systemNumber & 74.14\reddown & 80.64\reddown & \cellcolor{colorbest}\color{colortext}{93.38}$\tbspace$ & \cellcolor{colorbest}\color{colortext}{96.21}$\tbspace$ & \cellcolor{colorsecond}\color{colortext}95.79$\tbspace$ & \cellcolor{colorsecond}\color{colortext}97.50\greenup & 72.82\reddown & 78.81\reddown & \cellcolor{colorsecond}\color{colortext}87.62$\tbspace$ & 92.41$\tbspace$ & \cellcolor{colorsecond}\color{colortext}92.68$\tbspace$ & \cellcolor{colorsecond}\color{colortext}95.84$\tbspace$ \\
    \system-UB   & 93.55$\tbspace$ & 96.00$\tbspace$ & 89.72$\tbspace$ & 93.91$\tbspace$ & 94.78$\tbspace$ & 96.24$\tbspace$ & 91.46$\tbspace$ & 94.86$\tbspace$ & 87.53$\tbspace$ & \cellcolor{colorsecond}\color{colortext}92.58$\tbspace$ & 90.96$\tbspace$ & 93.97$\tbspace$ \\
    \system      & \cellcolor{colorsecond}\color{colortext}95.54$\tbspace$ & \cellcolor{colorbest}\color{colortext}{97.21}\greenup & \cellcolor{colorsecond}\color{colortext}90.98$\tbspace$ & \cellcolor{colorsecond}\color{colortext}94.68$\tbspace$ & \cellcolor{colorbest}\color{colortext}{96.40}\greenup & \cellcolor{colorbest}\color{colortext}{97.77}\greenup  & \cellcolor{colorsecond}\color{colortext}92.88$\tbspace$ & \cellcolor{colorsecond}\color{colortext}95.85$\tbspace$ & \cellcolor{colorbest}\color{colortext}{89.05}$\tbspace$ & \cellcolor{colorbest}\color{colortext}{93.58}$\tbspace$ & \cellcolor{colorbest}\color{colortext}{93.19}$\tbspace$ & \cellcolor{colorbest}\color{colortext}{96.09}$\tbspace$ \\

    \hline    
    \end{tabular}
    \label{tab:comparison_jbeil_pivoting}
\end{table*}

\subsection{Ablation Study}
\label{app:ablation}

\ignore{
\begin{table}[h]
    \centering
    \caption{KL scores on different client numbers.\zl{since we talk about OpTC regarding client number, we also need its result}\jxu{done}\zl{emmm, i don't think it can explain why LANL 2 and 5 are bad, how about computing the node number variance? }\jxu{Done}
    }    
    \begin{tabular}{c|ccccc}
    \hline
    Dataset & 2 & 3 & 4 & 5  \\
    \hline
    OpTC & 3.30 & 3.91 & 2.94 & - \\    
    LANL & 2.04 & 2.24 & 2.47 & 2.38\\
    \hline
    \end{tabular}
    \label{tab:kl_client}
\end{table}
}

\zztitle{Impact of ACS.} 
We introduce ACS to dynamically adjust the clients' weights of each iteration. Here we measure its contributions by comparing the AP and AUC of \system with and without it. We found that ACS increases AP and AUC in nearly every combination of GNIDS, dataset, and $K$. Taking AP as an example, the increases averaged among $K$ are 7.76\% for OpTC+Euler, 0.24\% for LANL+Euler, 14.63\% (Transductive) and 16.68\% (Inductive) for LANL+Jbeil, and 6.53\% (Transductive) and 7.53\% (Inductive) for Pivoting+Jbeil.

In Figure~\ref{fig:weight_update4}, we illustrate the changes of client weight along FL iterations, and we show the result of LANL+Euler when $K=4$ due to space limit. It turns out ACS is able to notably adjust the client weights, e.g., Client 3 and 4 after Epoch 15, which addresses the issue of gradient instability when updating the global model.

\begin{figure}
    \centering
    \includegraphics[width=1.0\linewidth]{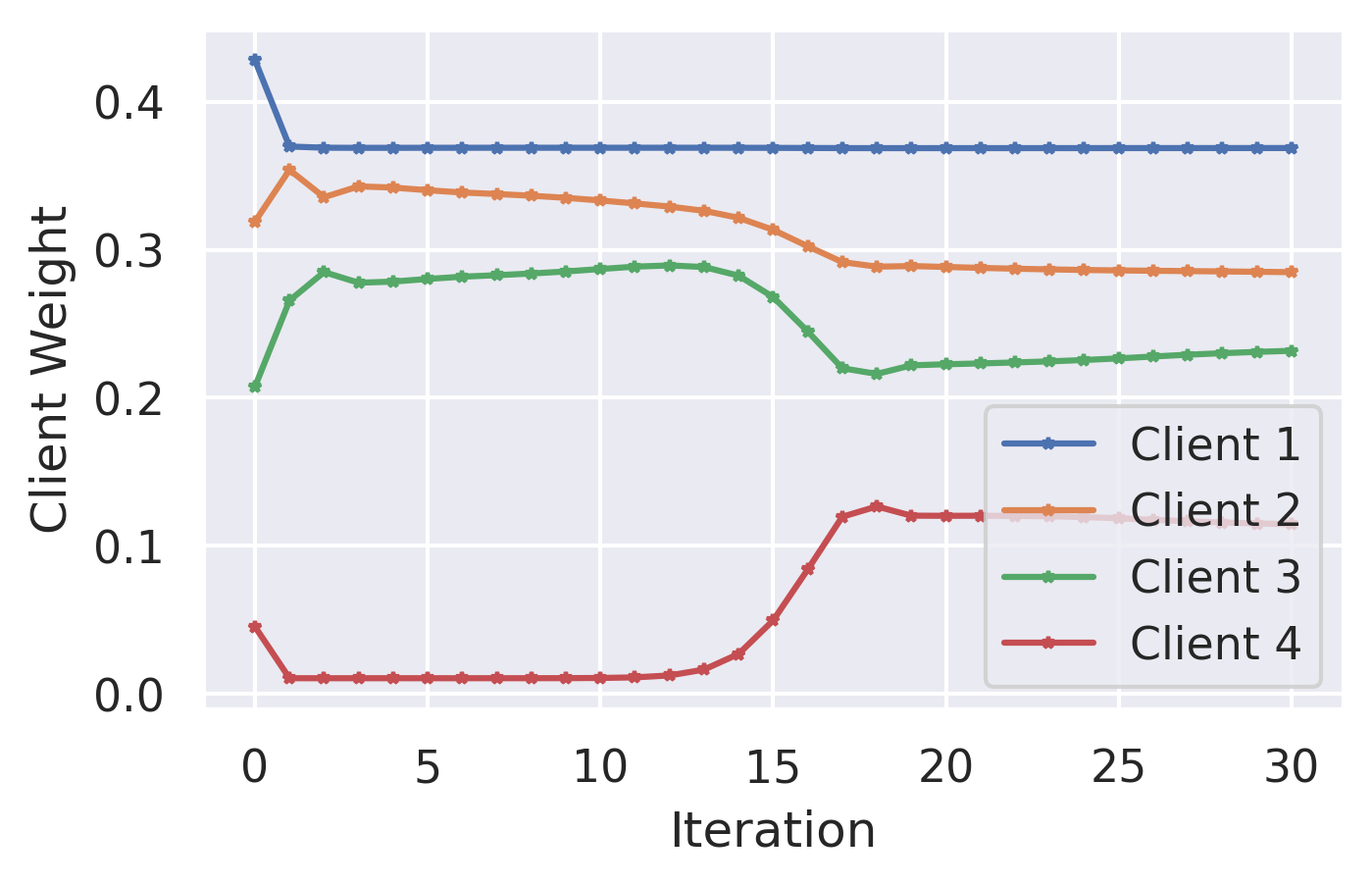}
    \caption{The client weight updating on LANL ($K=4$). 
    }
    \label{fig:weight_update4}
\end{figure}

\zztitle{Impact of graph augmentation.}
We propose to augment each client's subgraph with 1-hop cross-client edges as described in Section~\ref{subsec:local_graph}. Here we evaluate the contribution of this technique. We rerun the experiments for OpTC ($K=3$) and LANL ($K=4$)  when Euler is the GNIDS without adding the cross-client edges. The result is shown in Table~\ref{tab:graph_aug}. It shows prominent improvement with graph augmentation: over 82\% and 56\% increase in AP and AUC on OpTC is achieved, and 0.61\% and 0.7\% increase in AP and AUC on LANL is achieved.

\begin{table}[h]
    \caption{Analysis of graph augmentation. ``1-hop'' is the default setting. ``None'' means the cross-client edges are not added. Euler is the GNIDS.}
    \centering
    \begin{tabular}{l|c|cc}
    \hline
        Dataset & Augmentation & AP(\%) & AUC(\%) \\
        \hline
        \multirow{2}{*}{OpTC ($K=3$)} & None & 1.01 & 43.80 \\
        & 1-hop & 83.29 & 99.76 \\   
        \hline        
        \multirow{2}{*}{LANL ($K=4$)} & None & 0.11 & 96.23  \\
        & 1-hop & 0.72 & 97.00   \\
        \hline
    \end{tabular}
    \label{tab:graph_aug}
\end{table}

\revisednew{\zztitle{Clients' weight initialization.}
We leverage BA model to initialize the clients' weight to address the issue of data heterogeneity. Here we evaluate the contribution of this component.
Table~\ref{tab:ablation_bamodel} compares the setting with BA model enabled and disabled (i.e., equal initial weights for all clients). Results at early iterations (3 and 5) show that the BA model helps \system achieve higher AP and AUC scores, resulting in faster convergence and better early-stage learning performance.
}
\begin{table}[]
    \centering
    \caption{\revisednew{The ablation study of the clients' weight initialization. ``Equal'' means an equal initial weight for each client.}}    
    \begin{tabular}{lc|cc|cc}
    \hline
        & &\multicolumn{2}{c|}{Iteration=3} & \multicolumn{2}{c}{Iteration=5} \\
        Dataset & K & AP (\%) & AUC (\%) & AP (\%) & AUC (\%) \\
        \hline
        OPTC (equal) & 3     & 27.28 & 94.12 & 47.61 & 96.26 \\
        OPTC (\system) & 3 & 28.12 & 97.61 & 73.95 & 97.88 \\
        LANL (equal)  & 4   & 0.46  & 96.75 & 0.61  & 97.10 \\
        LANL (\system) & 4 & 0.85  & 97.21 & 0.64  & 98.80 \\
    \hline
    \end{tabular}
    \label{tab:ablation_bamodel}
\end{table}

\zztitle{Client number ($K$).} 
Here, we take a closer look at the effectiveness results under different  $K$. We found that \system, FedOpt and FedProx are less sensitive to different $K$ values, suggesting their robustness under different client settings. The changes on FedAvg can be drastic, e.g., AP dropping from 75.19\% ($K=5)$ to 33.69\% ($K=2$) for OpTC+Euler, and from 88.54\% ($K=2$) to 50.94\% ($K=4$) for LANL+Jbeil+Transductive. Though \systemNumber sometimes can outperform \system, it is not as stable as \system, e.g., AP dropping from 93.38\% to 74.14\% for Pivoting+Jbeil+Transductive.

\begin{table}[h]
    \centering
    \caption{Standard deviation of node and event numbers among different clients under different $K$. ``Non-FL'' shows the numbers when GNIDS is trained on the whole dataset. ``+M'' and ``+K'' mean million events and thousand events.  }
    \begin{tabular}{c|ccc|ccc}
    \hline
    & \multicolumn{3}{c|}{\textbf{Node}} & \multicolumn{3}{c}{\textbf{Events}} \\
    & OpTC & LANL & Pivoting & OpTC & LANL & Pivoting \\
    \hline
    Non-FL & 814 & 17,649 & 1,015 & 92+M & 1,051+M & 74+M \\
    2      & 475 & 10,858 & 112   & 864+K     & 33+K        & 36+M \\
    3      & 408 & 5,626  & 139   & 759+K     & 5+M     & 34+M \\
    4      & 359 & 5,139  & 107   & 673+K     & 6+M     & 31+M \\
    5      & 324 & 4,721  & -     & 609+K     & 6+M     & - \\
    \hline
    \end{tabular}
    \label{tab:kl_client}
\end{table}

\begin{figure}[h]
    \centering

    \includegraphics[width=0.45\textwidth]{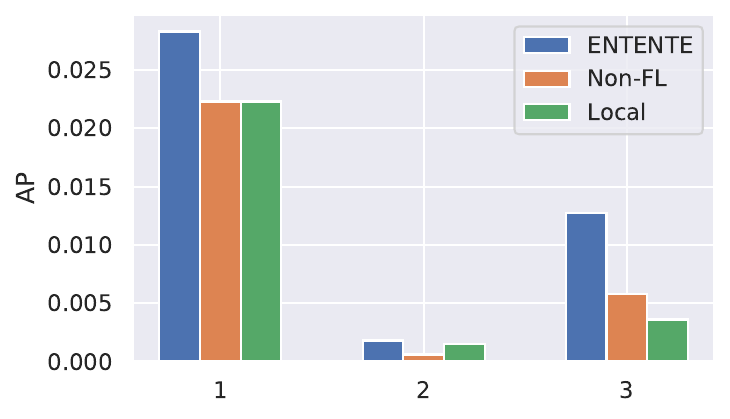}
    \includegraphics[width=0.45\textwidth]{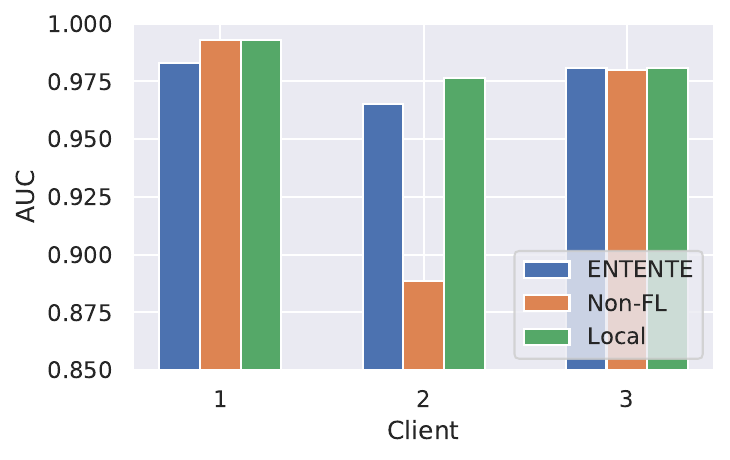}
    \caption{The AP and AUC of \system's global model, non-FL model and local model on each local dataset. LANL+Euler is tested and $K=3$.
    }
    \label{fig:client_lanl}
\end{figure}

We also examine the client clusters, by measuring the standard deviation of the node numbers and event numbers across clients. The result is shown in Table~\ref{tab:kl_client}. For LANL, the standard deviation of node numbers is the highest when $K=2$, which is over 40\% of the total node numbers, and it is significantly reduced for when $K>2$. \system turns out to achieve stable AP and AUC across different $K$, especially when Euler is the GNIDS, suggesting \system can self-adjust to different client distributions. 
For OpTC, we also observe a high standard deviation of node numbers. For Pioviting, though the standard deviation of node numbers is relatively small, we observe a very high standard deviation of event numbers. Overall, the data distribution reflects the non-iid challenge discussed in Section~\ref{subsec:goals}.

\ignore{Naseri et al. proposed a ``non-IIDness score'' to measure the client's data distribution~\cite{naseri2022cerberus}, and we also borrow this idea to explain our results. Naseri et al. compute the pair-wise KL divergence~\cite{kullback1951information} on the class histograms between different clients and then take the average. However, our training data only has one class, benign edges. So we make adjustments to compute the non-IIDness score on the WLH generated from each client's graph (WLH described in Algorithm~\ref{alg:wlh}). The result is listed at Table~\ref{tab:kl_client}. \zl{jason, i'll write it after seeing your result}
}

\zztitle{The number of initial edges for a new node ($m$).} 
We use the BA model to generate the reference graph. We set $m$ as 5 by default. Decreasing and increasing $m$ could make the reference graph denser or sparser, which could impact the weights computed for each client. To assess its impact, we set $m=2, 4, 5, 10, 100$ and compute AUC and AP on LANL with Euler GNIDS ($K=3$). 
We find the result does not change, suggesting \system{} is not sensitive to the value of $m$.
\ignore{
\begin{table}[h]
    \centering
    \caption{Different m weight calculation methods on LANL. The client number is 4. \jxu{different m does not impact the performance. tried m=2, 4, 5, 10, 100.}}
    \begin{tabular}{lcccccccccc}
    \hline
        m & AP & AUC & Precision & Recall & FPR\\
        \hline
        3 & \\
        4 & \\
        5 &  0.0627 & 0.9697 & 0.0044 & 0.9014 & 0.0074 \\
        6 & \\
        \hline
    \end{tabular}
    \label{tab:m_lanl}
\end{table}
}

\zztitle{Comparison between local and global models.} 
Each client obtains a global model from the central server after FL. On the other hand, the client can just train a \textit{local} model on its own data and ignore the global model, when it considers the global model does not bring benefit. As such, we compare \system's global model with the trained local model on each client. For a fair comparison, the cross-client edges have been added to each client's graph when training the local model. We also test how non-FL performs on each client's data separately. 

In Figure~\ref{fig:client_lanl}, we show the comparison of \system's global model, Non-FL, and local model on 3 clients of LANL,. For AP, \system outperforms the other methods on all clients (2.83\%, 0.18\%, 1.27\%). 
Client 1 and client 3 have a higher AP margin because many malicious edges (47.41\% and 40.79\%) reside in both clients. The result shows that for a client,  aggregating information from other clients is beneficial.

\zztitle{Cluster generation.}
We leverage MBM to cluster machines of a log dataset and form clients. Here we provide a visualized example of the clustering result on LANL in Figure~\ref{fig:lanl_network} with the MBM code~\cite{mbm}. We set the number of machine clusters, which equals to the client number $K$, to 4.
MBM is able to generate machine clusters of variable sizes (ranging from 129 to 15,941) and preserve the prominent interactions between user clusters and machine clusters (e.g., host cluster 1 connects to user cluster 0, 1 and 2).

\begin{figure*}
    \centering
    \includegraphics[width=1.0\linewidth]{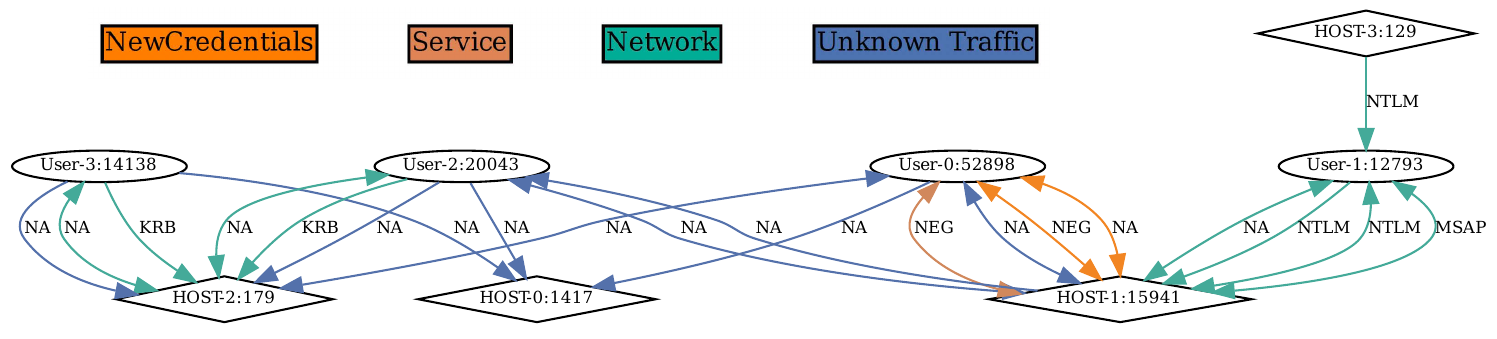}
    \caption{Visualization of the clustering done by MBM for LANL ($K=4$). We use ellipses to show user clusters and the numbers in the ellipses are the user numbers. For hosts (or machines), we use diamonds and provide host numbers.
    The legend in different colors summarizes the major traffic type between user clusters and machine clusters: New Credentials, Service, Network and Unknown Traffic. 
    Each edge represents one protocol type.
    ``KRB'' means Kerberos; ``NEG'' means Negotiate; ``MSAP'' means Microsoft Authentication Package. ``'NA'' means unknown protocols. 
    }
    \label{fig:lanl_network}
\end{figure*}

\begin{table}[h]
    \centering
    \caption{Defenses against strongest poisoning ($p=100\%$, $\gamma=100$) on LANL ($K=4$).    
    Client 4 is malicious. ``-'' in ``Defense'' means no defense is enabled. 
    }
    \begin{tabular}{llccc}
    \hline
    Type & Defense & AP(\%) & AUC(\%) & SR(\%) \\
    \hline
    \multirow{3}{*}{No attack} & - & 0.72 & 97.00 & - \\
    & Weak DP & 0.06 & 95.23 & - \\ %
    & CDP & 0.54 & 96.36 & - \\ %
    \hline
    \multirow{2}{*}{With attack} & Weak DP  & 0.11 & 95.21 & 2.77  \\
    & CDP & 0.54 & 96.37 & 9.10 \\
    \hline
    \end{tabular}
    \label{tab:defense_lanl}
\end{table}

\begin{table}[h]
    \centering
    \caption{Defenses against strongest poisoning ($p=100\%$, $\gamma=100$) on OpTC ($K=3$). 
    Client 3 is malicious.
    }
    \begin{tabular}{llccccccccc}
    \hline
    Type & Defense & AP(\%) & AUC(\%) & SR(\%) \\
    \hline
    \multirow{4}{*}{No attack} & - & 83.29 & 99.76 & - \\    
    & Weak DP & 74.75 & 96.66 & - \\ %
    & CDP & 24.04 & 86.94 & - \\ %
    \hline
    \multirow{3}{*}{With attack}  & Weak DP  & 74.01 & 96.66 & 16.25  \\
    & CDP & 24.03 & 86.95 & 38.66 \\
    \hline
    \end{tabular}
    \label{tab:defense_optc}
\end{table}

\subsection{Differential Privacy for \system}
\label{app:dpfl}

Specifically, Equation~\ref{eq_update_bounding} can be adjusted to integrate DP as written below:
\begin{equation}
\label{eq_dp}
    w_{i+1} =w_{i} +\frac{1}{K}\sum^{K}_{k=1} r^k_{k}\times   \mathsf{NB}(\Delta w^{k}_{i+1}) + \mathcal{N}(0,M_{\mathsf{qs}}^2\sigma^2)
\end{equation}
We are able to support weak DP~\cite{sun2019can} and CDP~\cite{geyer2017differentially,icml/ZhangCH0Y22} defenses with Equation~\ref{eq_dp}.
\revisednew{For theoretical completeness, \system realizes the DP guarantee parameterized by $\epsilon =O(K{M_{qs} \sqrt{R \log (1/\delta)}}/{\sigma})$, which is proved in Appendix~\ref{app:proof}.
\begin{theorem}
\label{the:dpEntente}
    Consider that $K$ clients collaboratively train a model in \system for $R$ rounds. Through Equation~\ref{eq_dp}, \system realizes $(K{M_{qs} \sqrt{R \log (1/\delta)}}/{\sigma}, \delta)$-DP.
\end{theorem}
} 
Their difference is the value of $M_{\mathsf{qs}}$.
Weak DP extends norm bounding by adding a Gaussian noise under small variance $\sigma$ to each model update, and the value of $M_{\mathsf{qs}}$ is relatively smaller than expected as the standard DP.
CDP requires the noise level to be proportional to the sensitivity.

We set $\omega$, $M$, $\epsilon$ and $\delta$ to 5.0, 5.0, 1.0, 0.1 and $\sigma$ as 1.0 for Weak DP and 0.2 for CDP after empirical analysis. 
The results of LANL and OpTC from Table~\ref{tab:defense_lanl} and Table~\ref{tab:defense_optc} show that though both defenses are effective against our poisoning attack, with Weak DP more effective in dropping SR, they also significantly drop GNIDS's AP even when no attack is launched (e.g., 0.0072 dropped to 0.0006 for LANL+Weak DP and 0.8329 dropped to 0.2404 for OpTC+CDP).

\subsection{Theoretical Analysis}
\label{app:proof}
\subsubsection{Proof of Theorem~\ref{the:diff_from_fedavg}}
\label{app:diff_from_fedavg}

\begin{proof}
    With Equation~\ref{equ:w_update}, we have, 
    \begin{equation*}
    \begin{aligned}
        | w_{i+1}| &= |\sum^{K}_{k=1} (c_1 \times \jsequ^k + c_2 \times S^k_i \times D^k_i) \times w^k_i|\\
        &\leq \sum^{K}_{k=1} |(c_1 \times \jsequ^k + c_2 \times S^k_i \times D^k_i) \times w^k_i|\\
        &\leq \sum^{K}_{k=1} |(c_1 \times \jsequ^k + c_2 \times S^k_i \times D^k_i) |\times| w^k_i|\\
    \end{aligned} 
    \end{equation*}
Given Equation~\ref{eq_dis}, we know that $D^k_i= \omega$ if $ \sqrt{\sum(w_i^k-w_{i-1})^2}\geq\omega$; otherwise $D^k_i= \sqrt{\sum(w_i^k-w_{i-1})^2}$ if $ \sqrt{\sum(w_i^k-w_{i-1})^2}<\omega$.
Thus, we bound $D_1^k  \leq \omega$ for any $\omega>0$. 
In practice, $\omega=1$ can be done by normalization.

 Jaccard similarity $\jsequ^k$ and Cosine similarity $S^k_i$ satisfy  $0\leq \jsequ^k,S^k_i\leq1$.
 For a preset positive integer $K$, we can simplify the inequality above and get,
\begin{equation*}
    \begin{aligned}
        | w_{i+1}| 
        &\leq \frac{1}{K}\sum^{K}_{k=1} |(c_1 \times 1 + c_2 \times 1 \times \omega) |\times| w^k_i|\\
        &\leq \sum^{K}_{k=1} \frac{|(c_1  + c_2 \omega) |\times| w^k_i|}{K}\\
 &= |(c_1  + c_2 \omega)| \frac{\sum^{K}_{k=1} | w^k_i|}{K}\\  
 \Rightarrow |\frac{w_{i+1}}{(1/K)\sum_k w^k_i}| &\leq c_1  + c_2 \omega\\
    \end{aligned} 
    \end{equation*}

\end{proof}

\subsubsection{Proof of Theorem~\ref{the:convergence}}
\label{app:nb_proof}
In \system, at each round $i\in [R]$, we solve an optimization problem in the following form,
\begin{equation}
    \min_{x\in \Rbb^d} f(w_{i+1})=\frac{1}{K}\sum_{k=1}^{K} F_k(w_i)
\end{equation}
where $k\in [K]$ be a client index for $K$ clients. 
For the $k$-th client, we define its loss function by,
\begin{equation}
   F_k(w)= \Ebb_{z\sim \mathrm{data}_k}[f_k(w,z)] 
\end{equation}
where  $\mathrm{data}_k$ is the client's data distribution.
For different clients, $\mathrm{data}_k$ can be heterogeneous. 
Recall in FedAVG, $F_k(w_i)=w_i$.

At each round $i$, we assume all $K$ clients participate in and contribute to updating the global model.
Following the notations used in Algorithm~\ref{alg:system}, \( E \) is the number of local epochs; \( \eta \) is the learning rate used in local models;  \( R \) is the number of maximum FL iterations; $c_1$ and $c_2$ are the predetermined hyperparameters to adjust contributors; $d$ is the dimension of $w$.

Before analyzing the convergence of \system, we first introduce the prior knowledge~\cite{iclr/ReddiCZGRKKM21,icml/ZhangCH0Y22} and common settings~\cite{nips/ZaheerRSKK18,icml/ReddiHSPS16,corr/KingmaB14} that we rely on. 
We adopt the standard assumption in Lemma~\ref{ass:lip_bound}~\cite{nips/ZaheerRSKK18,icml/ReddiHSPS16,corr/KingmaB14} for nonconvex optimization as $F_k$ may be nonconvex.
$\nabla F_k(x)$ is defined as the computed gradients of each client.
As for federated optimization, we consider the setting of bounded variance and gradients~\cite{mlsys/LiSZSTS20,iclr/ReddiCZGRKKM21,icml/ZhangCH0Y22} as in Lemmas~\ref{lem:bvar},\ref{lem:bvarg} and Lemmas~\ref{lem:bgra}, which are widely adopted for heterogeneity analysis. 

\begin{lemma}[Lipschitz Gradient]
\label{ass:lip_bound}
    The function $F_k$ is $L$-smooth for any $k\in[K]$ such that, 
    \begin{equation}
        \|\nabla F_k(x)- \nabla F_k(y)\| \leq L\| x-y\|,\quad \forall\  x, y \in \Rbb^d
    \end{equation}
\end{lemma}

\begin{lemma}[Bounded Local Variance] 
\label{lem:bvar}
The function $F_k$ has $\sigma_{\local}$-bounded local variance such that,
\begin{equation}
    \Ebb[\|\nabla [f_k(w,z)]_j-[\nabla F_k(w)]_j \|]=\sigma_{\local,j}^2
\end{equation}
 for all $w\in\Rbb^d, j\in [d]$, and $k\in[K]$. 
\end{lemma}

\begin{lemma}[Bounded Global Variance] 
\label{lem:bvarg}
The function $F_k$ has $\sigma_{\globalsf}$-bounded global variance such that,
\begin{equation}
    \frac{1}{K} \sum_{k=1}^{K}\|\nabla [F_k(w)]_j-[\nabla f(w)]_j\|\leq \sigma_{\globalsf,j}^2
\end{equation}
    for all $w\in\Rbb^d$ and $j\in [d]$.  
\end{lemma}

\begin{lemma}[Bounded Gradients]
    \label{lem:bgra}
    The gradients of function $f_k(w,z)$ is $G$-bounded such that,
    \begin{equation}
        |[\nabla f_k(w,z)]_j\leq G|, \quad \forall\  j\in[d]
    \end{equation}
for any $k\in[K], w\in\Rbb^d$ and $z\sim \mathrm{data}_k$.
\end{lemma}

\begin{proof}
    To analyze the convergence of \system, we need to express the global model difference $w_{i+1}-w_i$  between any consecutive rounds $i,i+1$.
    By Equation~\ref{eq_update_bounding}, the global model difference can be calculated by,
    \begin{equation*}
        w_{i+1}-w_i = \frac{1}{K}\sum^{K}_{k=1} (c_1 \times \jsequ^k + c_2 \times S^k_i \times D^k_i) \times  \mathsf{NB}(\Delta w^{k}_{i+1})
    \end{equation*}

At the $i$-th iteration, the learning rate of updating $w_i+1$ is $c_1 \times \jsequ + c_2 \times S^k_i \times D^k_i$ at each client.
Using Lipschitz smoothness, we get,
    \begin{equation*}\label{equ:main_equ}
    \begin{aligned}
       & \Ebb[f(w_{i+1})]\leq f(w_t) \\&+\left\langle\nabla f(w_t),\Ebb[\frac{1}{K}\sum^{K}_{k=1} (c_1 \jsequ^k + c_2  S^k_i D^k_i)  \mathsf{NB}(\Delta w^{k}_{i+1})]\right\rangle\\
        &+\frac{L}{2}\Ebb\left[\|\frac{1}{K}\sum^{K}_{k=1} (c_1  \jsequ^k + c_2  S^k_i  D^k_i)  \mathsf{NB}(\Delta w^{k}_{i+1})\|\right]\\
    \end{aligned}
    \end{equation*}
In Theorem~\ref{the:diff_from_fedavg}, we know $r^k_i=c_1 \times \jsequ^k + c_2 \times S^k_i \times D^k_i\leq c_1+c_2\omega$ for any positive $c_1,c_2$.
Norm bounding is defined by $\mathsf{NB}(\Delta w^{k}_{i+1})=\frac{\Delta w^{k}_{i+1}}{\max(1, \lVert \Delta w^{k}_{i+1} \rVert_2 / M)}$~\cite{sun2019can}. 
The bounding ensures that if $\| \Delta w^{k}_{i+1} \|_2\leq M$, $w^{k}_{i+1}$ is preserved for further computation; otherwise, if $\| \Delta w^{k}_{i+1} \|_2> M$, the norm of $\Delta w^{k}_{i+1}$ is equal to $M$.
For simplifying notations, we define,
\begin{equation*}
    \alpha^k_i=\frac{(c_1 \times \jsequ^k + c_2 \times S^k_i \times D^k_i)M}{\max (M, \eta\|\sum_{e=1}^{E}g^{k,e}_{i}\|)}
\end{equation*}
where $g^{k,e}_{i}$ is the stochastic gradient computed by client $k$ at the $e$-th local epoch. 
Then, we define $\Tilde{\alpha}^k_i$ to take the math expectation for all possible random variables in local epochs,
\begin{equation*}
    \Tilde{\alpha}^k_i=\frac{(c_1 \times \jsequ^k + c_2 \times S^k_i \times D^k_i)M}{\max (M, \eta\|\Ebb [\sum_{e=1}^{E}g^{k,e}_{i}]\|)}
\end{equation*}
 We replace with $r_i^k$ to focus on the analysis of norm bounding. That is, Inequality~\ref{equ:main_equ} can be simplified to be,
    \begin{equation*}
    \begin{aligned}
        \Ebb[f(w_{i+1})]&\leq f(w_t) \\
        &+\left\langle\nabla f(w_t),\Ebb[\frac{1}{K}\sum^{K}_{k=1} r_i^k \times  \mathsf{NB}(\Delta w^{k}_{i+1})]\right\rangle\\
        &+\frac{L}{2}\Ebb\left[\|\frac{1}{K}\sum^{K}_{k=1}r_i^k \times  \mathsf{NB}(\Delta w^{k}_{i+1})\|\right]\\
    \end{aligned}
    \end{equation*}
Using $\alpha_i^k,\Tilde{\alpha}^k_i$, the updated gradient $r_i^k \times  \mathsf{NB}(\Delta w^{k}_{i+1})$  can be calculated by,
\begin{equation*}
\begin{aligned}
    r_i^k  & \times\mathsf{NB}(\Delta w^{k}_{i+1}) = -\eta\sum_{e=1}^{E}g_i^{k,e}\cdot \alpha^k_i\\
    =&-\eta\sum_{e=1}^{E}g_i^{k,e}\cdot (\alpha^k_i-\Tilde{\alpha}^k_i+\Tilde{\alpha}^k_i)\\
    &=\left(-\eta\sum_{e=1}^{E}g_i^{k,e}\cdot (\alpha^k_i-\Tilde{\alpha}^k_i)\right)+ \left(-\eta\sum_{e=1}^{E}g_i^{k,e}\cdot \Tilde{\alpha}^k_i\right)
\end{aligned}
\end{equation*}
Let $\overline{\alpha}_i=\frac{1}{K}\sum_{k=1}^K\Tilde{\alpha}^k_i$. 
As for the first-order term in  Inequality~\ref{equ:main_equ}, we thus drive the following relations,
\begin{equation*}
\begin{aligned}
    &\left\langle\nabla f(w_i),\Ebb\left[\frac{1}{K}\sum^{K}_{k=1} r_i^k \times  \mathsf{NB}(\Delta w^{k}_{i+1})\right]\right\rangle\\
    =&\left\langle\nabla f(w_i),\Ebb\left[\frac{1}{K}\sum^{K}_{k=1}  \left(-\eta\sum_{e=1}^{E}g_i^{k,e}\cdot (\alpha^k_i-\Tilde{\alpha}^k_i)\right)\right]\right\rangle \\&+\left\langle\nabla f(w_i),\Ebb\left[\frac{1}{K}\sum^{K}_{k=1} \left(-\eta\sum_{e=1}^{E}g_i^{k,e}\cdot (\Tilde{\alpha}^k_i-\overline{\alpha}_i)\right)\right]\right\rangle 
    \\&+\left\langle\nabla f(w_i),\Ebb\left[\frac{1}{K}\sum^{K}_{k=1} \left(-\eta\sum_{e=1}^{E}g_i^{k,e}\cdot \overline{\alpha}_i\right)\right]\right\rangle   \\
\end{aligned}
\end{equation*}
 The first two terms are essentially statistical bias. 
 Thus, we reduce to the third term to analyze the bound of the first-order term.
 By taking math expectation over all $K$ clients and defining $\Ebb[g_i^{k,e}]=\nabla F_k(w^{k,e}_i)$, we have,
 
\begin{equation*}
\begin{aligned}
    &\Ebb\left[
    \left\langle\nabla f(w_i),\frac{1}{K}\sum^{K}_{k=1} \left(-\eta\sum_{e=1}^{E}g_i^{k,e}\cdot \overline{\alpha}_i\right)\right\rangle \right]\\
    &=\Ebb\left[
    \left\langle\nabla f(w_i),\frac{1}{K}\sum^{K}_{k=1} \left(-\eta\sum_{e=1}^{E}(g_i^{k,e}-\nabla F_k(w^{k,e}_i))\cdot \overline{\alpha}_i\right)\right\rangle \right]\\ 
    &\quad+ \Ebb\left[
    \left\langle\nabla f(w_i),\frac{1}{K}\sum^{K}_{k=1} \left(-\eta\sum_{e=1}^{E}\nabla F_k(w^{k,e}_i)\cdot \overline{\alpha}_i\right)\right\rangle \right]\\
\end{aligned}
\end{equation*}
Now, we reduce to the last term to analyze the convergence of the
first-order term, which is equivalent to the first-order analysis in  \cite[Theorem~3.1]{icml/ZhangCH0Y22}.
Combine with Theorem~\ref{the:diff_from_fedavg},  Theorem~3.1 \cite{icml/ZhangCH0Y22} and Lemma~3 in \cite{iclr/ReddiCZGRKKM21}, we get,
\begin{equation*}
    \begin{aligned}
       &\left\langle\nabla f(w_t),\Ebb[\frac{1}{K}\sum^{K}_{k=1} r_i^k \times  \mathsf{NB}(\Delta w^{k}_{i+1})]\right\rangle\\
    &\leq  5L^2E^2\eta^2\sigma_{\local}^2 (c_1+c_2\omega)\\
    &\quad +30L^2 E^3\eta^2(c_1+c_2\omega)(\sigma^2_{\globalsf}+\|\nabla f(w_i))\|^2)  \\
    \end{aligned}
\end{equation*}
where $\sigma^2_{\local}=\sum_{j=1}^{d} \sigma^2_{\local,j}$ and $\sigma^2_{\globalsf}=\sum_{j=1}^{d} \sigma^2_{\globalsf,j}$.
The analysis of the second-order term is similar to the first-order term. 
According to  Theorem~3.1 in \cite{icml/ZhangCH0Y22} and  Theorem~1 in \cite{iclr/ReddiCZGRKKM21}, we have, 
\begin{equation*}
\label{equ:second-order}
    \begin{aligned}
      & \frac{L}{2}\Ebb\left[\|\frac{1}{K}\sum^{K}_{k=1}r_i^k \times  \mathsf{NB}(\Delta w^{k}_{i+1})\|\right]\\
       & \leq 5\sigma^2_{\local}K\eta^4 \overline{\alpha}_iL^3E^2(c_1+c_2\omega)\\
      &\quad +2\sigma^2_{\globalsf}K\eta^2\overline{\alpha}_iE^3(15L^3\eta^2+L)(c_1+c_2\omega)\\
      &\quad +2K\eta^2\overline{\alpha}_iE^3(15L^3\eta^2+L)(c_1+c_2\omega)\|\nabla f(w_i)\|^2\\
    \end{aligned}
\end{equation*}
Combining all equations above /and constrain~\cite{iclr/ReddiCZGRKKM21}
  $$\eta\leq\min\{\frac{\sqrt{K}}{(c_1+c_2\omega)\sqrt{48E^3}}, \frac{K}{6EL(c_1+c_2\omega)(K-1)},\frac{1}{\sqrt{60}EL}\}$$, 
we have,

\begin{equation*}
\begin{aligned}
    &\Ebb[f(w_{i+1})]\leq f(w_{i}) - \frac{(c_1+c_2\omega)\eta\overline{\alpha}_iE}{4}\|\nabla f(w_i)\|^2\\
&+\left(\frac{5(c_1+c_2\omega)\eta^3\overline{\alpha}_i}{2}(1+\frac{12\eta (c_1+c_2\omega)}{K})L^2E^2\right)\sigma^2_{\local}\\
&+\left(\frac{3L}{K}(c_1+c_2\omega)^2\eta^2\overline{\alpha}_i^2E\right)\sigma^2_{\local}\\
&+ \frac{30(c_1+c_2\omega)\eta^3\overline{\alpha}_i}{2}(1+\frac{12\eta (c_1+c_2\omega)}{K})L^2E^3\sigma_{\globalsf}^2\\
& + \left\langle \nabla f(w_i),\Ebb\left[\frac{1}{K}\sum^{K}_{k=1}  \left(-\eta\sum_{e=1}^{E}g_i^{k,e}\cdot (\alpha^k_i-\Tilde{\alpha}^k_i)\right) \right]\right\rangle\\
&+ {\left\langle\nabla f(w_i),\Ebb\left[\frac{1}{K}\sum^{K}_{k=1} \left(-\eta\sum_{e=1}^{E}g_i^{k,e}\cdot (\Tilde{\alpha}^k_i-\overline{\alpha}_i)\right)\right]\right\rangle }\\
&+\frac{3L}{2}(c_1+c_2\omega)^2\Ebb\left[\left\|\frac{1}{K}\sum^{K}_{k=1}  \left(-\eta\sum_{e=1}^{E}g_i^{k,e}\cdot (\alpha^k_i-\Tilde{\alpha}^k_i)\right)\right\|^2\right]\\
&+\frac{3L}{2}(c_1+c_2\omega)^2\Ebb\left[\left\|\frac{1}{K}\sum^{K}_{k=1}  \left(-\eta\sum_{e=1}^{E}g_i^{k,e}\cdot (\alpha^k_i-\Tilde{\alpha}^k_i)\right)\right\|^2\right]\\
\end{aligned}
\end{equation*}
After federated training $R$ rounds, we get the expectation over gradients norm bounding, 
\begin{equation*}
\begin{aligned}
    &\frac{1}{R}\sum_{i=1}^{R}\Ebb[\overline{\alpha}_1\|\nabla f(w_i)\|^2]\\
    &\leq\underbrace{\frac{4}{(c_1+c_2\omega)R\eta E}\left(\Ebb\left[f(w_1)\right]-\Ebb\left[f(w_{R+1})\right]\right)}_{O(1/{(R\eta E(c_1+c_2\omega))})} + G\\
    & + \underbrace{\frac{10\eta^2L^2EK+120(c_1+c_2\omega)\eta^3L^2E}{K}\sigma_{\local}^2(\frac{1}{R}\sum_{i=1}^{R}\overline{\alpha_i})}_{<O((c_1+c_2\omega)\eta/K*\sigma^2_{\local})}\\
    &+\underbrace{\frac{12(c_1+c_2\omega)L}{K}\eta \sigma_{\local}^2(\frac{1}{R}\sum_{i=1}^{R}(\overline{\alpha_i})^2)}_{O((c_1+c_2\omega)\eta/K\cdot\sigma^2_{\local})}\\
    & + \underbrace{\frac{60E(\eta^2L^2EK+12(c_1+c_2\omega)\eta^3L^2E)}{K}\sigma_{\globalsf}^2(\frac{1}{R}\sum_{i=1}^{R}\overline{\alpha_i})}_{O(\eta^2E^2\cdot\sigma^2_{\globalsf})}\
\end{aligned}
\end{equation*}
If the DP-style noise is added, the upper bound additionally involves $O((c_1+c_2\omega)d/(\eta EK)\cdot\sigma_{\mathsf{DP}})$, where $\sigma_{\mathsf{DP}}=M_{\mathsf{qs}}\sigma$ in Equation~\ref{eq_dp}.
\end{proof}

\revisednew{
\subsubsection{Proof of Theorem~\ref{the:dpEntente}}
The goal is to show the entire system  of $\system$ achieves $(K\frac{M_{\mathsf{qs}}}{\sigma}\sqrt{R\log(1/\delta)},\delta)$-DP guarantee.
Each client, in every round, adds independent Gaussian noise with variance ($\sigma^2$) to their sensitive data. By the classic Gaussian mechanism, adding this amount of noise ensures each output satisfies (($\epsilon_0, \delta$))-DP~\cite{corr/abs-2210-00597} for single-shot queries, where ($\epsilon_0 = O\left(\frac{M_{qs}}{\sigma}\sqrt{\log(1/\delta)}\right)$). Here, ($M_{qs}$) captures the local sensitivity due to quantization and sparsification.
Since $K$ clients participate in each round and the adversary could potentially gain access to all outputs, we use the basic composition theorem~\cite{corr/abs-2210-00597}. 
Accordingly, the total privacy loss across $K$ parallel mechanisms is bounded by $K\epsilon_0$, so for $K$ clients,
$\epsilon_1=K\epsilon_0 = O\left(\frac{KM_{\mathsf{qs}}}{\sigma}\sqrt{\log(1/\delta)}\right)$.
This step relies on the composition property of DP, and assumes clients' outputs are independent.
 When the protocol proceeds for $R$ rounds, we need to further compose the privacy loss using the advanced composition theorem~\cite{corr/abs-2210-00597}, which more tightly bounds the cumulative privacy. 
 This theorem states that, when composing $R$ mechanisms each giving $(\epsilon_1, \delta)$-DP, the overall privacy is $\epsilon = O\left(\epsilon_1\sqrt{R\log(1/\delta)}\right)$.
Plugging in from above, we get
$\epsilon = O\left(\frac{KM_{\mathsf{qs}}}{\sigma}\sqrt{R\log(1/\delta)}\right)$, which is what Theorem 3 claims.
 }

\end{document}